\documentclass[journal]{IEEEtran}

\usepackage{algorithm}
\usepackage{algorithmicx}
\usepackage{amsfonts}
\usepackage{cite,graphicx,amsmath,amsthm}
\usepackage{subfigure}
\usepackage{fancyhdr}
\usepackage{dsfont}
\usepackage{array,color}
\usepackage{bm}
\usepackage{amssymb}
\usepackage{mathrsfs}
\usepackage{cases}
\usepackage{amsmath}
\usepackage{amsthm}
\usepackage{threeparttable}
\usepackage{multirow}
\usepackage{epstopdf}
\usepackage{booktabs}
\usepackage{setspace}
\usepackage{soul}

\usepackage{epstopdf}

\usepackage{algpseudocode}
\allowdisplaybreaks[4]

\newtheorem{lemma}{Lemma}
\newtheorem{proposition}{Proposition}
\newtheorem{theorem}{Theorem}

\begin{document}
\title{Edge Learning with Unmanned Ground Vehicle:\\ Joint Path, Energy and Sample Size Planning}

\author{Dan Liu, Shuai Wang, Zhigang Wen, Lei Cheng, Miaowen Wen and Yik-Chung Wu
\thanks{
This work was supported in part by the National Key R\&D Program of China under Grant No.2019YFF0302601,
in part by the Shenzhen Fundamental Research Program under Grant JCYJ20190809142403596,
in part by the Guangdong Basic and Applied Basic Research Foundation under Grant 2019A1515111140,
in part by the Fundamental Research Funds for the Central Universities under Grant 2019SJ02, and in part by the Open Research Fund from Shenzhen Research Institute of Big Data under Grant 2019ORF01012. (\textit{Corresponding Author: Shuai Wang}).

Dan Liu and Zhigang Wen are with the Beijing Key Laboratory of Work Safety Intelligent Monitoring, Department of Electronic Engineering, Beijing University of Posts and Telecommunications, Beijing 100876, China (e-mail: dandanmessage@bupt.edu.cn, zwen@bupt.edu.cn).

Shuai Wang is with the Department of Electrical and Electronic Engineering, Southern University of Science and Technology, Shenzhen 518055, China, and is also with the Department of Computer Science and Technology, Southern University of Science and Technology, Shenzhen 518055 (e-mail: wangs3@sustech.edu.cn).

Lei Cheng is with the Shenzhen Research Institute of Big Data, Shenzhen, Guangdong, P. R. China (e-mail: leicheng@sribd.cn).

Miaowen Wen is with the School of Electronic and Information Engineering, South China University of Technology, Guangzhou 510641, China (e-mail:
eemwwen@scut.edu.cn).

Yik-Chung Wu is with the Department of Electrical and Electronic Engineering, The University of Hong Kong, Hong Kong (e-mail: ycwu@eee.hku.hk).
}
}
\maketitle

\begin{abstract}
Edge learning (EL), which uses edge computing as a platform to execute machine learning algorithms, is able to fully exploit the massive sensing data generated by Internet of Things (IoT).
However, due to the limited transmit power at IoT devices, collecting the sensing data in EL systems is a challenging task.
To address this challenge, this paper proposes to integrate unmanned ground vehicle (UGV) with EL.
With such a scheme, the UGV could improve the communication quality by approaching various IoT devices.
However, different devices may transmit different data for different machine learning jobs and a fundamental question is how to jointly plan the UGV path, the devices' energy consumption, and the number of samples for different jobs?
This paper further proposes a graph-based path planning model, a network energy consumption model and a sample size planning model that characterizes F-measure as a function of the minority class sample size.
With these models, the joint path, energy and sample size planning (JPESP) problem is formulated as a large-scale mixed integer nonlinear programming (MINLP) problem, which is nontrivial to solve due to the high-dimensional discontinuous variables related to UGV movement.
To this end, it is proved that each IoT device should be served only once along the path, thus the problem dimension is significantly reduced.
Furthermore, to handle the discontinuous variables, a tabu search (TS) based algorithm is derived, which converges in expectation to the optimal solution to the JPESP problem.
Simulation results under different task scenarios show that our optimization schemes outperform the fixed EL and the full path EL schemes.
\end{abstract}

\begin{IEEEkeywords}
Edge learning, Internet of Things, mixed integer nonlinear programming, unmanned ground vehicle.
\end{IEEEkeywords}

\IEEEpeerreviewmaketitle

\section{Introduction}

Edge learning (EL) is an emerging intelligent system that integrates ubiquitous sensing (e.g., camera, Lidar), wireless communication (e.g., cellular network, Wi-Fi), and machine learning (e.g., regression, classification) \cite{el1}. In contrast to single-machine learning systems \cite{ml}, \cite{ml1} {where data needs to be centralized, EL does not have such drawbacks as sensing devices can share resources with each other.
With a growing number of participants, the learning performance of EL can be improved.
Furthermore, the EL system uses edge computing as a platform to train and execute deep learning algorithms.
Therefore, the enormous Internet of Things (IoT) data and computing resources in the network become closer, which significantly reduces the communication latency and loading to other part of the network.
Based on the above reasons, EL is envisioned as a necessary component in the sixth generation (6G) telecommunication system \cite{6g1}, and is believed to accelerate the convergence of many relevant research fields including industry electronics, unmanned systems, wireless communications, machine learning, and automatic control.

\subsection{Background}

The development of EL is mainly based on the following techniques:
\begin{itemize}
\item[1)] \textit{Internet of Things.} The IoT is envisioned as a world-wide network infrastructure composed of numerous embedded devices with sensors and electronics \cite{iot1}, \cite{iot}. In the IoT, the devices are generally wireless-interconnected, and it is expected that the number of connected IoT devices will exceed 50 billion by 2020 \cite{iotdevice}. The gigantic quantities of IoT devices will generate zillions Bytes of data, which consequently leads to a key challenge: The generated massive data require intensive computation and energy resources for signal processing, while the IoT devices are computation and energy limited \cite{iot3}.

\item[2)] \textit{Edge Computing.} To address the challenges in IoT, edge computing (EC) deploys computing servers at the edge of the network, which is in close proximity to the data source. The vast available resources in the edge servers can be leveraged to provide elastic data collection, storage, and processing power to support capability-constrained IoT devices \cite{ec2,ec3}. Compared to the cloud computing paradigm, EC delivers several benefits, including shortened service delay, reduced communication overhead, as well as increased security, and is suitable for the scenarios involving delay-sensitive applications \cite{ec4}, \cite{ec5}.

 \item[3)] \textit{Deep Learning.} Recently, deep learning has been applied to many IoT applications, since it can extract insights for high-quality decision making and trend prediction with the large-scale data generated from enormous amount of IoT devices \cite{dl2,dl3}. In particular, with the growing ubiquity of camera-enabled IoT devices, a huge volume of images are being produced everyday, and convolutional neural networks (CNNs) have been considered as the dominant strategy \cite{cnn2,cnn3}.
However, the CNN models in existing research works are mostly deployed at the cloud, and it is necessary to push the CNN models towards the edge for IoT applications.
\end{itemize}

\subsection{Challenges of EL and Proposed EL-UGV}

In EL systems, there are three major challenges.
\begin{itemize}
\item[1)] \textit{Optimizing for Learning Performance.} In traditional wireless communications, all data are treated equally, and the objective is to maximize the system throughput under a variety of budget constraints (e.g., power, time, bandwidth constraints). In contrast, wireless communication in EL should maximize the learning performance. Therefore, it is necessary to distinguish the value of data in the context of machine learning and transmit high-valued data with priority.

\item[2)] \textit{Pathloss in Wireless Communications.}
In cloud learning systems, different devices are connected via Internet. In contrast, the devices in EL are connected through wireless. Due to signal attenuation during propagation of wireless signals, the performance of EL will be limited by the pathloss and noises of wireless channels.

\item[3)] \textit{Limited Energy at IoT Devices.} Due to the huge volume and the small size, most IoT devices cannot employ large batteries. Therefore, their transmit powers are in the order of mW or even uW, making it difficult to collect their data from far-away base stations.
\end{itemize}

In order to address the above challenges, this paper proposes a novel EL framework that integrates unmanned ground vehicles (UGVs) into EL. The proposed framework is therefore termed EL with UGV (EL-UGV).
The EL-UGV consists of: 1) a sample size planning module that collects more samples for difficult learning jobs and fewer samples for easy jobs; 2) a path planning module that selects the target stopping points and the corresponding moving routes; 3) an energy planning module that automatically finds the best trade-off between moving energy and communicating energy.
The three modules together contribute to a joint path, energy and sample size planning (JPESP) optimization problem that maximizes the learning performance of EL.

However, since the collected data samples could be imbalanced among different classes, existing learning accuracy model in \cite{ao1}, \cite{ao2} is not applicable.
To maximize the learning performance under imbalanced data, the proposed sample size model characterizes F-measure as a function of the minority class sample size.
On the other hand, a graph-based path planning model and a network energy consumption model are used to describe the mobility of UGV and the transmission costs at IoT.
With the proposed models, the JPESP problem turns out to be a large-scale mixed integer nonlinear programming (MINLP) problem, which is nontrivial to solve due to the high-dimensional discontinuous variables related to UGV movement.
To this end, it is first proved that each IoT device should be served only once along the path, thus allowing us to reduce the problem dimension by a factor of $J$ (i.e., the number of UGV stopping points).
Furthermore, to handle the discontinuous variables, a tabu search (TS) based algorithm is derived, which converges to the optimal solution of the JPESP problem.
Simulation results under different task scenarios verify our optimization scheme and show that the performance outperforms other benchmark schemes such as the fixed EL and the full path EL schemes.
To sum up, the main contributions of the paper are listed as follows.
\begin{itemize}
\item First, we consider a novel UGV-enabled mobile EL communication (EL-UGV) system, which supports multiple deep learning tasks and accesses the IoT devices via time division multiple access (TDMA) protocol.
To the best of the authors' knowledge, currently there is no research work focusing on the combination of UGV and EL.
\item Next, considering the collected data samples are usually imbalanced, we use F-measure instead of accuracy as the performance metric to evaluate the EL performance.
Moreover, a novel learning performance model that characterizes F-measure as a function of the minority class sample size is proposed.
\item
  Based on the proposed learning model, we then establish a JPESP problem that is maximizing the minimum F-measure for all tasks by jointly optimizing the UGV path, transmission time, transmission power, and the minimum number of samples among all classes, subject to constraints of communication capacity, total execution time, total energy consumption and the graph mobility.
\item
    Finally, to solve the challenging large-scale MINLP problem, we prove that each IoT device should be served only once along the path and derive the TS-based algorithm to optimize the UGV path.
    Simulation results are presented to validate our analysis.
\end{itemize}

\subsection{Outline}
The rest of this paper is organized as follows. Section II presents the related works on path planning, energy planning, and sample size planning. Section III introduces the system model for the EL-UGV system, which consists of path planning model, energy planning model, and sample size planning model. Section IV proposes the F-measure model and its curve fitting procedure. Section V formulates the JPESP problem, derives the TS-based path design algorithm, and considers the practical implementation. Simulation results are presented in Section~VI and conclusions are drawn in Section VII.

\emph{Notation}. Symbol notations are summarized in Table I.

\begin{table*}[ht]
\caption{Summary of Symbol Notations} 
\centering 
\begin{tabular}{|l|l|l|}
\hline
\textbf{Variable} & \textbf{Description} \\
\hline
$s_{j}\in\{0,1\}$  & $s_{j}=1$ represents the vertex $j$ appears in the routing path; $s_{j}=0$ otherwise.\\
$E_{{r_{j},r_{i}}}\in\{0,1\}$  & $E_{{r_{j},r_{i}}}=1$ represents the edge $(r_{j},r_{i})$ being involved in the path; $E_{{r_{j},r_{i}}}=0$ otherwise.\\
$t_{u,j}\in [0,T_j]$ & Time (in $\mathrm{s}$) allocated to device $u$ when UGV is stopping at the vertex $j$.  \\
$p_{u,j}\in [0,P_{max}]$ & Transmit power (in $\mathrm{Watt}$) allocated to device $u$ when UGV is stopping at the vertex $j$. \\
\hline
\textbf{Parameter} & \textbf{Description} \\
\hline
$U$ & Number of IoT devices. \\
$R$ & Number of steps. \\
$M$ & Number of tasks. \\
$C_m$ & Number of classes of image samples for task $m$.  \\
$A$ & The data amount (in $bits$) for each image sample. \\
$(\mathbf{w},\mathbf{b})$ & Parameters of CNN model, where $\mathbf{w}$ is the weights and $\mathbf{b}$ is the biases. \\
$\mathcal{J}$ & Set of $J$ vertices standing for the possible stopping points. \\
$\mathcal{L}$ & Set of directed edges standing for the feasible movement routes. \\
$\mathcal{R}$ & Routing path. \\
$D_{i,j}$ & Distance between vertex $i$ and vertex $j$. \\
$(\gamma_{1},\gamma_{2})$ & Parameters related to the weight of UGV. \\
$v$ & Constant speed (in $\mathrm{m/s}$) of the UGV. \\
$T_j$ & Duration when UGV stops at vertex $j$. \\
$P_{max}$ & Maximum transmit power at devices. \\
$\epsilon\in \left(0,1\right)$ & Hyper-parameter that balances the energy consumption between UGV and devices. \\
$h_{u,j}$ & Uplink channel from device $u$ to UGV. \\
$\mathcal{G}_{c,m}$ & The set of all devices transmitting the samples for class $c$ of task $m$. \\
$({\theta}_{1,m},{\theta}_{2,m},{\theta}_{3,m})$ & Parameters to be determined by the data sets and CNN model. \\
$T_{all}$ & Completion time (in $\mathrm{s}$) including UGV moving and data transmission. \\
$E_{all}$ & Energy budget (in $\mathrm{Joule}$). \\
${\sigma}^{2}$ & Receiver additive white Gaussian noise (AWGN) power noise power (in $\mathrm{Watt}$). \\
\hline
\end{tabular}
\end{table*}

\section{Related Works}
\subsection{Path Planning}
For UGV \cite{ugv1}, due to the limitation of energy resources and the hostility of wireless channel, its path greatly influences the efficiency of data gathering and consequently affects the system performance. While the path planning problem is extremely difficult to solve as it is NP-hard \cite{path1,path2,path3}, fortunately, the meta-heuristic search algorithm is a promising way to handle the problem \cite{tmh2}, and several meta-heuristic algorithm based methods have been presented to design UGV's optimal path so far. Adam \emph{et al.} \cite{up1} proposed a modified particle swarm optimizer (PSO) to heuristically solve the Solar-Powered UGV motion planning problem that minimizes the travel time while considering the net energy constraint. Mohammad \emph{et al.} \cite{up2} developed an ant colony optimization strategy called green ant for UGV path planning that jointly considers the UGV energy consumption. You-Chiun \emph{et al.} \cite{up3} used a local search based efficient path planning algorithm for mobile sink to realize reliable data gathering. However, no methodology of UGV's path planning has yet been considered in EL systems.

\subsection{Energy Planning}
An important issue for wireless communication systems is energy planning. So far researchers have applied joint time and power optimization approaches and energy minimization techniques. For example, Qingqing \emph{et al.} \cite{enp1} jointly optimized time allocation and power control for maximization of the weighted sum of the user energy efficiencies for wireless powered communications system. Following this, for the same system, Qingqing \emph{et al.} \cite{enp2} maximized the weighted sum rate of the devices in the consideration of downlink and uplink time allocation and the power control at the devices. Bin \emph{et al.} \cite{enp3} proposed a joint time scheduling and power allocation scheme to optimize the sum-rate for relay assisted batteryless IoT Networks. Qizhong \emph{et al.} \cite{enp4} minimized overall downlink and uplink energy consumption in WET-enabled Networks. Thang X \emph{et al.} \cite{enp5} considered a energy minimization technique for cache-assisted content delivery networks with wireless backhaul. However, there have been no detailed studies about energy planning for UGV-based EL, taking into account the motion energy and communication energy.

\subsection{Sample Size Planning}
For sample size planning, the overall/average accuracy (OA/AA) is commonly used to assess the effect of the dataset size on learning performance. For example, Junghwan \emph{et al} \cite{ao3} employed the AA to exam learning performance at given training sample size in CNN-based medical image classification systems. However, in practice the collected data are typically imbalanced (i.e., the number of training samples in one class may be much less than that in another), and OA/AA is not suitable for this case since it neglects the minority class and is strongly biased towards the majority class, which leads to misleading conclusions \cite{nonao1}, \cite{nonao2}. F-measure is a more favorable and meaningful performance measure than OA/AA for the imbalanced data and has been widely used in classification performance of CNN models in recent works. For example, Guanbin \emph{et al}. \cite{fm1} employed F-measure to evaluate the performance of the proposed multiscale deep CNN features based visual saliency model. Minghui \emph {et al}. \cite{fm2} investigated a CNN-based end-to-end trainable fast scene text detector, i.e., TextBoxes++, and used F-measure as the performance measure to verify its accuracy and efficiency. Prashant \emph{et al}. \cite{fm3} proposed a compact end-to-end deep CNN for moving object detection named MSFgNet, and used F-measure to evaluate its performance.  However, very limited research has been focused on F-measure in EL system.

\section{System Model}
\begin{figure*}
\centering
\includegraphics[width=182mm]{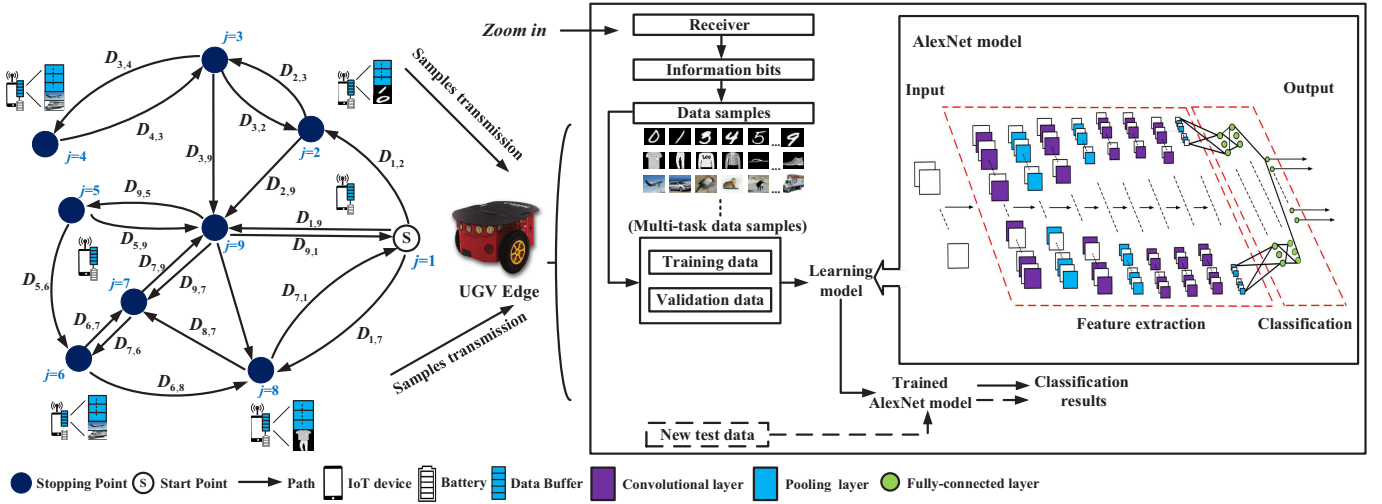}
\caption{The EL-UGV system model.}
\label{fig:}
\end{figure*}

We consider an UGV-enabled mobile EL communication (EL-UGV) system shown in Fig.~1, which consists of a surveillance UGV and $U$ IoT devices.
In particular, the goal of the UGV is to train $m$ classification models for object recognition by collecting $C_m$ classes of image samples observed at the $U$ devices (e.g., camera sensors).
It is assumed that each device transmits one class of images for one learning task and the samples from each device are independently and identically distributed.
In order to achieve a satisfactory learning performance, deep learning technique (e.g., AlexNet
 shown in Fig. 1), which uses multiple layers of nonlinear processing units such that the end-to-end model matches the data very well, is adopted.
However, the samples for tuning the deep learning models' parameters are generated from IoT devices that could be far away from the UGV.
Therefore, the UGV needs to complete two steps: 1) move around to approach various devices and 2) collect samples to learn from the data.
Below, we present the details of the two steps.
\subsection{Path Planning Model}

To model the mobility of UGV, a directed graph ($\mathcal{J}$,$\mathcal{L}$) is adopted, where $\mathcal{J}=\left\{1,...,J \right\}$ denotes the set of $J$ vertices standing for the possible stopping points, and $\mathcal{L}$ denotes the set of directed edges standing for the feasible movement routes.
Based on this graph, the UGV needs to determine a routing path $\mathcal{R}=(r_1,...,r_R)$ with $r_i\in\mathcal{J}$ for any $i=1,\cdots,R$ and $(r_i,r_{i+1})\in\mathcal{L}$ for any $i=1,\cdots,R-1$, where $R$ stands for the number of steps to be taken.
This path can be equivalently described by a vertex selection variable $\mathbf{s}=[s_1,\cdots,s_J]^T$ and a path selection matrix
\begin{align}
&\mathbf{E}=
\left[
\begin{matrix}
&E_{1,1} & \cdots & E_{1,J} &\\
& \cdots  & \cdots & \cdots &\\
&E_{J,1} & \cdots & E_{J,J} &\\
\end{matrix}
\right],
\end{align}
where we use $s_j=1$ if vertex $j$ appears in the $\mathcal{R}$ and $s_j=0$ otherwise, and ${E}_{r_{i},r_{i+1}}=1$ with $i=1,\cdots,R-1$ and zero otherwise.
The vertex section variable and the path selection matrix should satisfy $\mathbf{E}\in \mathcal{Q}(\mathbf{s})$, where
\begin{subequations}
\begin{align}
\mathcal{Q}(\mathbf{s})=\Bigg\{\mathbf{E}:&
\sum_{i=1}^JE_{{r_{j},r_{i}}}=\sum_{i=1}^JE_{r_{i},r_{j}}=s_{j},\,\forall j, \label{2a}\\
&{\eta}_i-{\eta}_j+\left(\sum_{r=1}^Js_r-1\right)E_{{r_{j},r_{i}}}
\nonumber\\
&+\left(\sum_{r=1}^Js_r-3\right)E_{{r_{i},r_{j}}}
\nonumber\\
&
\leq
\sum_{r=1}^Js_r-2+J_0(2-s_i-s_j),\nonumber\\
&\forall {i,j}\geq 2, \, i \neq j, \label{2b}\\
&s_{j}\leq {\eta}_j\leq \left(\sum_{r=1}^Ps_r-1\right)s_{j}, \, \forall j\geq 2, \label{2c}\\
&E_{{r_{j},r_{i}}}\in\left\{0,1\right\}, \, \forall {i,j}, \quad E_{{r_{j},r_{j}}}=0, \, \forall j, \label{2d}\\
&s_1=1, \, s_j\in\left\{0,1\right\}, \, \forall {j}\geq 2 \label{2e} \Bigg\}.
\end{align}
\end{subequations}

The constraint (2a) guarantees that each vertex in the routing path must have one edge pointing toward it and the other edge pointing away from it. The constraints (2b) and (2c) are subtour elimination constraints, which guarantee that the path must be connected. The constraint (2d) represents whether edge $(r_{j},r_{i})$ from candidate $r_{j}$ to $r_{i}$ belongs to the routing path, and constraint (2e) represents whether candidate vertex $j$ appears in the routing path. Notice that $\{{\eta}_j\}$ represent slack variables to guarantee a connected trajectory and $J_0$ is a constant set to $J_0=10^6$.

\subsection{Energy Planning Model}

Having specified the path planing model, the next step is to model the energy consumption, which consists of two parts, i.e., UGV motion energy and users' transmission energy.
To compute UGV motion energy, the moving distance is computed as $\mathrm{Tr}(\mathbf{D}^T\mathbf{E})$, where $D_{i,j}$ is the distance between vertex $i$ and vertex $j$\footnote{Notice that $D_{i,j}=0$ if $i=j $ and $D_{i,j}=+\infty$ if there is no feasible route between $i$ and $j$.}.
Therefore, the total motion energy in Joule is given by
\begin{equation}\label{6}
\Theta(\mathbf{E})=\underbrace{\mathrm{Tr}(\mathbf{D}^T\mathbf{E})}_{\mathrm{Distance}}\times
\left(\frac{\gamma_{1}}{v}+\gamma_{2}\right),
\end{equation}
where $v$ is the UGV speed in $\mathrm{m/s}$, and $(\gamma_{1},\gamma_{2})$ are parameters related to the weight of UGV (e.g., for the considered Pioneer 3DX robot, we have $\gamma_{1}=0.29$ and $\gamma_{2}=7.4$ \cite[Sec. IV-C]{3dx}.
Accordingly, the total motion time of UGV along path $\mathcal{R}$ can be expressed as
\begin{equation}\label{5}
t_{\mathrm{UGV}}(\mathbf{E})=\frac{\mathrm{Tr}(\mathbf{D}^T\mathbf{E})}{v}.
\end{equation}

On the other hand, without loss of generality, we suppose that the UGV stops at a particular vertex $j\in\mathcal{R}$ for a duration of $T_j$.
Out of this $T_j$, the UGV allocates $t_{u,j}\leq T_j$ to collect samples from device $u\in[1,U]$, and the transmit power at the IoT device $u$ for uploading its data is denoted as
$p_{{u},{j}}$ with $p_{{u},{j}}\leq P_{\mathrm{max}}$, where $P_{\mathrm{max}}$ is the maximum transmit power of the devices.
As a result, the total transmission energy at IoT is given by $\sum_{j=1}^J\sum_{u=1}^Up_{{u},{j}}t_{{u},{j}}$.
The total network energy is
\begin{align}
&\Psi\left(\mathbf{E},\{t_{u,j},p_{u,j}\}\right)=\sum_{j=1}^J\sum_{u=1}^Up_{{u},{j}}t_{{u},{j}}+\epsilon\,\Theta(\mathbf{E}),
\end{align}
where $\epsilon\in \left(0,1\right)$ is a hyper-parameter that weights the energy consumption at UGV and devices.

\subsection{Sample Size Planning Model}

Having specified the path and energy models, the next step is to model the sample size planning procedure.
In particular, the number of samples $x_{u,j}$ being transmitted by device $u$ at vertex $j$ is proportional to the transmission time, and is also determined by the quality of wireless channel for device $u$.
By adopting the Shannon capacity theorem, $x_{u,j}$ is upper bounded as:
\begin{align}\label{4}
&x_{u,j}\leq
t_{u,j}\times B\,\mathrm{log}_2\left(1+s_j\times\frac{p_{{u},{j}}\left|h_{{u},{j}} \right|^2}{{\sigma}^{2}}\right)
\Big/A,
\end{align}
where $B$ is the bandwidth available for the system, $h_{{u},{j}}\in\mathbb{C}$ is the uplink channel from device $u$ to the UGV, ${\sigma}^{2}$ is the AWGN noise power, and $A$ is the data amount in bits for each training sample.
Notice that the channel condition $\{h_{{u},{j}}\}$ can be pre-determined using ray-tracing methods \cite{ray}.

Based on the above sample size planning model, the total number of samples $y_{c,m}$ obtained for training the class $c$ of task $m$ is
\begin{align}\label{4}
y_{c,m}=\sum_{u\in\mathcal{G}_{c,m}}\sum_{j=1}^Jx_{u,j},
\end{align}
where $\mathcal{G}_{c,m}$ denotes the set of all devices transmitting the samples for class $c$ of task $m$, with $c=1,\cdots,C_m$ and $m=1,\cdots,M$.

\section{Modeling Learning Performance}

\subsection{F-Measure Model}
With the path planning, energy planning and sample size planning models, the remaining question is how to model the learning performance of classifiers.
While existing works mostly focus on the OA/AA, these metrics are not applicable to the considered UGV edge deep learning system.
This is because the OA/AA merely reflects the overall/average classification performance of all classes, and ignores the performance of the minority class (i.e., the class having the minimum number of training samples).
But since the UGV moves according to the planned path, it is close to a part of devices while far from the others.
As a result, the number of samples can vary significantly for different classes, leading to imbalanced learning performance.

To account for the imbalanced samples, with the confusion matrix shown in Table II presented in Appendix A, the F-measure $\Phi_m$ is defined as:
\begin{equation}
\Phi_m=\frac{2P_mR_m}{P_m+R_m}. \label{Fmeasure}
\end{equation}
where $P_m,R_m$ represents precision and recall, respectively, and they are given by
\begin{equation}\label{7}
P_m=\frac{\mathrm{TP}_m}{\mathrm{TP}_m+\mathrm{FP}_m},
\end{equation}
\begin{equation}\label{8}
R_m=\frac{\mathrm{TP}_m}{\mathrm{TP}_m+\mathrm{FN}_m}.
\end{equation}

Now, to model the relationship between the samples size $\{y_{c,m}\}$ and the F-measure, it is necessary to develop a F-measure model.
However, to the best of the authors' knowledge, currently there is no analytical model for predicting the F-measure given a certain dataset.
To this end, based on the definition in (8), it is observed that the F-measure has the following features:
\begin{itemize}

\item $\Phi_m$ is a monotonically increasing function of $\mathrm{TP}_m$, which is proportional to the minimum number of samples $\mathrm{min}_c\,y_{c,m}$ among all classes.

\item $\Phi_m$ is bounded as $0\leq \Phi_m\leq 1$ as it represents a percentage.

\item As $\mathrm{min}_c\,y_{c,m}$ increases, $\nabla\Phi_m$ would decrease and approach zero if $\mathrm{min}_c\,y_{c,m}\rightarrow+\infty$.
\end{itemize}

Based on the above features, this paper proposes a nonlinear F-measure model as:
\begin{equation}\label{model}
\Phi_{m}\left(\{y_{c,m}\}\right)=
{\theta}_{1,m}\times \left(\mathop{\mathrm{min}}_{c=1,\cdots,C_m}\,y_{c,m}\right)^{{\theta}_{2,m}}+{\theta}_{3,m},
\end{equation}
where $({\theta}_{1,m},{\theta}_{2,m},{\theta}_{3,m})$ are parameters to be determined by the data sets and deep leaning model. Particularly, ${\theta}_{2,m}$ is expected to be positive, parameters ${\theta}_{1,m}$ and parameters ${\theta}_{2,m}$ have the same sign, especially, when ${\theta}_{1,m}>0$, $0<{\theta}_{2,m}<1$.
Notice that in contrast to the learning accuracy model in \cite{ao3}, the proposed model in \eqref{model} adopts $\mathop{\mathrm{min}}_{c=1,\cdots,C_m}\,y_{c,m}$ as the input of function, thus allowing us to analyze the imbalanced data set.

\subsection{Curve Fitting of F-measure}

To verify the proposed model \eqref{model}, we use the AlexNet \cite{alex} to fulfill two deep learning tasks, and then fit the corresponding experimental data to the model \eqref{model}.
More specifically, our first task is to classify $10$ different categories (including airplane, birds, and etc.) in the CIFAR-10 dataset consisting of $60,000$ $32\times32$ color labeled images.
On the other hand, our second task is to classify $10$ different categories (including T-shirts, trousers, and etc.) in the Fashion-MNIST dataset consisting of $70,000$ $28\times28$ gray-scale labeled images.
Without loss of generality, the class ``airplane'' in CIFAR-10 and the class ``T-shirt'' in Fashion-MNIST are chosen as the minority class, respectively.
In particular, we consider the minority classes (i.e., airplane and T-shirt) ranges from $50$ to $2600$, while the number of samples for the other classes is set to $2600$.
In both experiments, for training the AlexNet model, we take $80\%$ of the data samples as training data to learn the parameters $(\mathbf{w},\mathbf{b})$, where $\mathbf{w}$ and $\mathbf{b}$ are weights and biases of the model respectively, and take the remaining $20\%$ as the validation data to finetune hyperparameters and guide the design of proper network architectures. Accordingly, we use a batch size of $32$, a weight decay of $0.0005$, a momentum rate of $0.9$, and a learning rate of $0.005$ for $10000$ iterations.
Furthermore, graphic processing units (GPUs) are used in order to speed up the training procedure.

Under the above settings, the F-measure versus the number of samples for the minority class (i.e., the red and blue solid line, and the red and blue points) is shown in Fig. 2 (a). The corresponding coefficients of the F-measure model for task 1 is given by ${\theta}_{1,1}=-3.742$, ${\theta}_{2,1}=-0.3957$ and ${\theta}_{3,1}=1.04$. And that for task 2 is given by ${\theta}_{1,2}=-0.9465$, ${\theta}_{2,2}=-0.3852$ and ${\theta}_{3,2}=0.955$.
It can be seen that the proposed nonlinear F-measure model (i.e. red and blue solid line) matches the experimental data (i.e., red and blue points) very well, no matter we use the CIFAR-10 or Fashion-MNIST dataset.
On the other hand, all curves increase significantly with the the number of samples for the minority class.
This corroborates the discussions in Section II-C that the F-measure is determined by the minority class instead of the average sample size (notice that all the other classes have fixed number of samples).

\begin{figure}[!t]
\centering
\includegraphics[width=85mm]{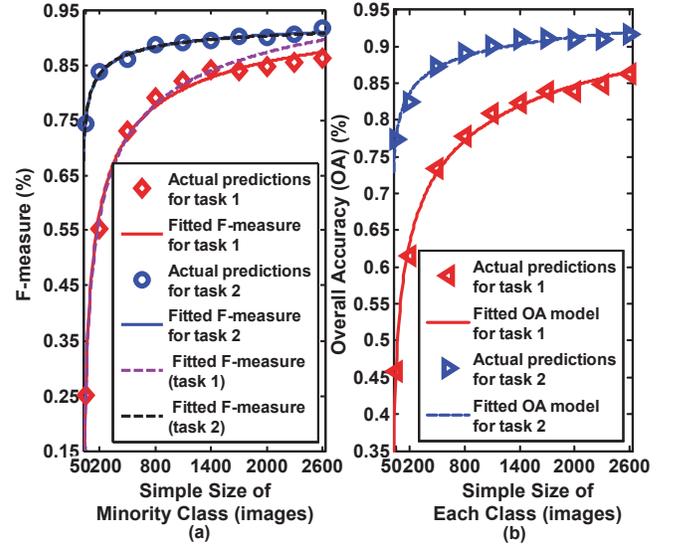}
\caption{The experimental results for the CIFAR-10 and Fashion-MNIST datasets. For both experiments, the test dataset contains 1000 untrained images, each class has 100 images. (a) is the F-measure for the CIFAR-10 and Fashion-MNIST datasets.
For extrapolation, the number of historical samples for the minority class varies from $50$ to $1700$. (b) is the experimental results of OA for the CIFAR-10 and Fashion-MNIST datasets.}
\label{fig_sim}
\end{figure}

\textbf{Remark.}
One may also adopt the OA model when the numbers of samples for different classes of task $m$ are the same:
\begin{equation}\label{12}
\Phi_m(\{y_{c,m}\})=\zeta_{1,m}y_{1,m}^{\zeta_{2,m}}+\zeta_{3,m},
\end{equation}
where $\{\zeta_{1,m},\zeta_{2,m},\zeta_{3,m}\}$ are parameters to be fitted. To verify the proposed model (12), we consider the same learning model and hyper-parameters in Section IV-B. By ranging the number of samples for each class from $50$ to $2600$, we compute the OA $\Xi_m$ as:
\begin{equation}
\Xi_m=\frac{\mathrm{TP}_m+\mathrm{TN}_m}{\mathrm{TP}_m+\mathrm{TN}_m+\mathrm{FP}_m+\mathrm{FN}_m}, \label{Fmeasure}
\end{equation}
where $\{\mathrm{TP}_m,\mathrm{TN}_m,\mathrm{FP}_m,\mathrm{FN}_m\}$ are defined in Table II. By fitting the model (12) to the experimental data, the parameters in the OA model are given by $(\zeta_{1,1},\zeta_{2,1},\zeta_{3,1})=(-1.64,-0.1803,1.263)$ for task 1, and $(\zeta_{1,2},\zeta_{2,2},\zeta_{3,2})=(-0.5914,-0.1927,1.05)$ for task 2. The OA versus the number of samples for each class is shown in Fig. 2 (b). It can be seen that the proposed OA model matches
the experimental data of both tasks very well.
\section{Joint Path, Energy and Sample Size Planning}
\subsection{Problem Formulation}
With the F-measure model, we establish a JPESP problem focusing on a max-min fairness F-measure design. The considered problem is to maximize the minimum F-measure $\Phi_1,...,\Phi_M$ for all tasks by planning the UGV path (including vertex selection variable $\mathbf{s}$ and selection matrix $\mathbf{E}$), scheduling the wireless resources (including transmission time $\{t_{u,j}\}$ and transmission power $\{p_{u,j}\}$), and the minimum number of samples among all classes $\{\mathrm{min}_cy_{c,m}\}$, subjecting to constraints of communication capacity, total execution time, total energy consumption and the graph mobility. To this end, an optimization problem is formulated as:
\begin{subequations}
\begin{align}
\mathrm{P}1:
\mathop{\mathrm{max}}\limits_{\substack{\mathbf{s},\mathbf{E},\{\alpha_{m}\}\\\{t_{{u},{j}},p_{{u},{j}}\}}}
\quad&\mathop{\mathrm{min}}_{\forall m}\Bigg({\theta}_{1,m}\alpha_{m}^{{\theta}_{2,m}}
+{\theta}_{3,m}\Bigg) \nonumber
\\\mathrm{s.t.}\quad\quad&
\alpha_{m}\leq\sum_{u\in\mathcal{G}_{c,m}}\sum_{j=1}^Jt_{u,j}\frac{B}{A}
\nonumber\\
&\mathrm{log}_2\left(1+s_j\frac{p_{{u},{j}}\left|h_{{u},{j}}\right|^2}{{\sigma}^{2}}\right),\quad \forall {c,m}, \label{P1a}
\\
&\mathbf{E}\in\mathcal{Q}(\mathbf{s}), \label{P1b}
\\
&\sum_{j=1}^J\sum_{u=1}^Ut_{{u},{j}}+\frac{\mathrm{Tr}(\mathbf{D}^T\mathbf{E})}{v}\leqslant T_{all}, \label{P1c}
\\
&\Psi\left(\mathbf{E},\{t_{u,j},p_{u,j}\}\right)\leqslant E_{all}, \label{P1d}
\\
&(1-s_{j})t_{{u},{j}}=0,\quad \forall {u,j}, \label{P1e}
\\
&t_{{u},{j}}\geq 0, \quad p_{{u},{j}}\geq 0,\quad \forall {u,j}, \label{resource}
\end{align}
\end{subequations}
where $\alpha_{m}=\mathrm{min}_c~y_{c,m}$.

\begin{itemize}
\item
The constraint \eqref{P1a} is the capacity limitation; it means that the number of training samples $\alpha_m$ should not be greater than the number of collected samples.
\item
The constraint \eqref{P1b} is the mobility management constraint; it describes the constraint of graph-based path \cite{graphconstraint}.
\item
The constraint \eqref{P1c} ensures that the execution including UGV moving and data transmission to be completed in $T_{all}$ seconds \cite{timeconstraint}.
\item
The constraint (14d) guarantees that the total energy consumption of the UGV and the devices cannot exceed the energy budget $E_{all}$ \cite{ugv1},\cite{timeconstraint}.
\item
The constraint \eqref{P1e} ensures that if the vertex $j$ is not visited, the transmission time $\{t_{{u},{j}}\}$ must be zero.
\item
The constraint (14f) is the non-negativity constraints on the optimization variables.
\end{itemize}

\indent However, problem $\mathrm{P}1$ is a large-scale MINLP problem, which is nontrivial to solve due to the discontinuity of $(\mathbf{s},\mathbf{E})$.
While a naive brute-force search can obtain the global optimal solution, it involves exponential computational complexities.
On the other hand, existing algorithms such as local search \cite{ls} may lead to local optimum or on plateaus where many solutions are equally fit.
To this end, in the following, we will derive an efficient tabu search (TS) algorithm based UGV's path planning method that converges to the optimal solution to $\mathrm{P}1$ with high probability.
\subsection{Tabu Search Algorithm}

To facilitate the derivation of TS-based algorithm, we need to reformulate $\mathrm{P}1$ into a compact form.
By defining
\begin{equation}\label{12}
F_{{u},{j}}=\frac{\left|h_{{u},{j}} \right|^2}{{\sigma}^{2}},
\end{equation}
and a function
\begin{align}
\Omega(\mathbf{s})=\Bigg
\{\mathop{\mathrm{min}}_{\forall m}\Bigg({\theta}_{1,m}\alpha_m^{{\theta}_{2,m}}
+{\theta}_{3,m}\Bigg):\eqref{P1a}-\eqref{resource} \Bigg\},  \label{Omega}
\end{align}
problem $\mathrm{P}1$ can be transformed into an equivalent problem only involving $\mathbf{s}$:
\begin{subequations}
\begin{align}
\mathrm{P}2:
\mathop{\mathrm{max}}\limits_{\mathbf{s}}~~
&\Omega(\mathbf{s})\nonumber\\
\mathrm{s.t.}~~~&s_1=1, \, s_j\in \{0,1\}, \, \forall j\geq 2
\end{align}
\end{subequations}

Now, in order to address the challenge brought by the discontinuity of $\mathbf{s}$ in $\mathrm{P}2$, this section will consider a meta-heuristic algorithm termed TS \cite{ts1}, which is based on local search and a tabu list $\mathcal{I}$.
The idea of TS is to avoid getting stuck at a local optimal point by maintaining a tabu list $\mathcal{I}$.
The tabu list $\mathcal{I}$ records the obtained local solutions that have tabu status and cannot be visited currently (but may be visited after a chosen number of iterations).
Moreover, the tabu restriction of a solution can be overridden if the solution passes an aspiration level $\mu$, where $\mu$ is the objective value of the best solution that has been found so far \cite{ts3}.

Notice that the maximum tabu list size $L$ of $\mathcal{I}$ will affect the convergence speed and performance of TS.
More specifically, a larger $L$ will prune more solutions, which results in slow convergence of TS.
On the other hand, a smaller $L$ will cause cycling of solutions, which results in suboptimal performance of TS.
Therefore, the $L$ is a hyper-parameter that needs to be tuned carefully (e.g., $L=10$ in this paper).
Finally, given the list size $L$, the list content is not fixed, meaning that it can either add or revoke an element based on the new solution produced in each iteration.

Based on the above descriptions, the TS algorithm for solving $\mathrm{P}2$ is as follows.
\begin{itemize}
\item \textbf{Step 1.} Generate an initial solution $\mathbf{s}_0$ (e.g., $\mathbf{s}_0=[1,0,\ldots,0]^T$), and set the current solution $\mathbf{s}_{c}$ and the best solution $\mathbf{s}^\star$ as $\mathbf{s}_{c}=\mathbf{s}^\star=\mathbf{s}_0$.

\item \textbf{Step 2.} Select a neighboring solution within $\mathcal{N}(\mathbf{s}_{c})$, which is the non-tabu solution yielding the highest objective function value $\Omega$ or the tabu solution passing the aspiration level $\mu$. To reduce the computational complexity, this step can be implemented via random sampling.

\item \textbf{Step 3.} Refresh the tabu list as well as $\mathbf{s}^\star$.

\item \textbf{Step 4.} Terminate the process if the maximum number of iterations $\rm{\overline{iter}}$ has been reached (or no improvement in $\mathbf{s}^\star$ has been found after a maximum number of iterations) and set optimal $\mathbf{s}^{*}=\mathbf{s}^\star$. Otherwise go to \textbf{Step 2}.
\end{itemize}
Notice that $\mathcal{N}(\mathbf{s}_{c})$ in \textbf{Step 2} is the neighborhood of $\mathbf{s}_{c}$.
Since $\mathbf{s}$ is a binary variable, a natural $\mathcal{N}(\mathbf{s}_{c})$ can be defined as:
\begin{equation}
\mathcal{N}(\mathbf{s}_{c})=\left\{\mathbf{s}:\|\mathbf{s}-\mathbf{s}_c\|_0\leq Z,\ \mathbf{s}\in\Pi\right\},
\end{equation}
where $Z$ is the size of its neighborhood and $\Pi$ is entire solution space.
The detailed description of the TS is shown in \textbf{Algorithm 1}.

\begin{algorithm}[!t]
\caption{Solve JPESP proplem via TS}
\begin{algorithmic}[1]
\State {\bf Initialize} $\rm{iter}=0$ and $\mathcal{I}=\emptyset$. Set $L=10$ and $Z=3$.
\State {\bf Initialize} $\mathbf{s}_{c}=\mathbf{s}^{\star}=\mathbf{s}_0=[1,0,\ldots,0]^T$.
\State {\bf Repeat}:
\State \quad Generate a set $\mathcal{A}$ with sufficiently large $|\mathcal{A}|$ and each element randomly sampled from $\mathcal{N}(\mathbf{s}_{c})$.
\State \quad {\bf For} all $\mathbf{x}^{[n]}$ in $\mathcal{A}$:
\State \quad\quad Solve $\mathrm{P}3$ with $\mathbf{s}=\mathbf{x}^{[n]}$ to obtain $\mathbf{E}^{[n]}$.
\State \quad\quad Solve $\mathrm{P}6$ with $\mathbf{s}=\mathbf{x}^{[n]}$ and $\mathbf{E}^{[n]}$ to obtain $\{\alpha^{[n]}_m,t^{[n]}_{u,j},p^{[n]}_{u,j}\}$.
\State \quad\quad Compute $\Omega(\mathbf{x}^{[n]})$ with $\mathbf{x}^{[n]}$, $\mathbf{E}^{[n]}$, $\{\alpha^{[n]}_m,t^{[n]}_{u,j},p^{[n]}_{u,j}\}$.
\State \quad Permute $\mathcal{A}$ into $\mathcal{B}=\{\mathbf{x}^{[1]},\cdots,\mathbf{x}^{[|\mathcal{A}|]}\}$ such that $\Omega(\mathbf{x}^{[1]})\geq \Omega(\mathbf{x}^{[2]})\geq \cdots \Omega(\mathbf{x}^{[|\mathcal{A}|]})$.
\State \quad {\bf For} $n=1$ to $|\mathcal{A}|$:
\State \quad \quad If $\Omega(\mathbf{x}^{[n]})>\Omega(\mathbf{s}^{\star})$: (\emph{Update Solution})
\State \quad \quad \quad Set $\mathbf{s}^{\star}\leftarrow \mathbf{x}^{[n]}$. (\emph{Update UGV visiting points})
\State \quad \quad \quad Set $\mathbf{E}^{\star}\leftarrow \mathbf{E}^{[n]}$. (\emph{Update UGV path})
\State \quad \quad \quad Set $\{\alpha^{\star}_m\leftarrow \alpha^{[n]}_m\}$. (\emph{Update sample size})
\State \quad \quad \quad Set $\{t^{\star}_{u,j}\leftarrow t^{[n]}_{u,j}\}$. (\emph{Update transmission time})
\State \quad \quad \quad Set $\{p^{\star}_{u,j}\leftarrow p^{[n]}_{u,j}\}$. (\emph{Update transmission energy})
\State \quad \quad {\bf If} $\mathbf{x}^{[n]}\notin\mathcal{I}$: (\emph{Update Tabu List})
\State \quad \quad \quad {\bf If} $|\mathcal{I}|<L$:
\State \quad \quad \quad \quad Update $\mathcal{I}=\mathcal{I}\bigcup\{\mathbf{x}^{[n]}\}$ and set $\mathbf{s}_{c}\leftarrow \mathbf{x}^{[n]}$.
\State \quad \quad \quad {\bf Else}:
\State \quad \quad \quad \quad Find $\mathbf{z}=\rm{argmin}_{\mathbf{s}\in\mathcal{I}}~\Omega(\mathbf{\mathbf{x}^{[n]}})$.
\State \quad \quad \quad \quad Update $\mathcal{I}=\mathcal{I}\setminus \{\mathbf{z}\}\bigcup\{\mathbf{x}^{[n]}\}$ and set $\mathbf{s}_{c}\leftarrow \mathbf{x}^{[n]}$.
\State \quad \quad \quad {\bf End If}
\State \quad \quad {\bf Else}:
\State \quad \quad \quad {\bf If} $\Omega(\mathbf{x}^{[n]})>\mu$:
\State \quad \quad \quad \quad Update $\mu=\Omega(\mathbf{x}^{[n]})$.
\State \quad \quad \quad \quad Update $\mathcal{I}=\mathcal{I}\setminus\{\mathbf{x}^{[n]}\}$ and set $\mathbf{s}_{c} \leftarrow \mathbf{x}^{[n]}$.
\State \quad \quad \quad {\bf End If}
\State \quad \quad {\bf End If}
\State \quad {\bf End For}
\State \quad Update $\rm{iter}\leftarrow \rm{iter}+1$.
\State {\bf Until}  $\rm{iter}=\rm{\overline{iter}}$.
\State {\bf Set} optimal $\mathbf{s}^{*}=\mathbf{s}^\star$, $\mathbf{E}^{*}=\mathbf{E}^\star$, $\{\alpha^{*}_m=\alpha^{\star}_m\}$, $\{t^{*}_{u,j}=t^{\star}_{u,j}\}$, and $\{p^{*}_{u,j}=p^{\star}_{u,j}\}$.
\State {\bf Output} $\mathbf{s}^{*}$, $\mathbf{E}^{*}$, $\{\alpha^{*}_m,t^{*}_{u,j},p^{*}_{u,j}\}$.
\end{algorithmic}
\end{algorithm}

\subsection{Computing $\Omega$ with Fixed $\mathbf{s}$}

To execute \textbf{Algorithm 1}, the remaining task is to compute $\Omega(\mathbf{s})$ for $\mathbf{s}=\mathbf{s}^{\diamond}$, where $\mathbf{s}^{\diamond}$ is any feasible solution to $\mathrm{P}2$.
Based on the expression of $\Omega(\mathbf{s})$ in \eqref{Omega}, computing $\Omega$ is equivalently to handling the following problem related to the path selection matrix $\mathbf{E}$ and the wireless resources $\left\{t_{{u},{j}},p_{{u},{j}}\right\}$:
\begin{subequations}
\begin{align}
\mathrm{P}3:
\mathop{\mathrm{max}}\limits_{\substack{\mathbf{E},\{\alpha_{m}\}\\\{t_{{u},{j}},p_{{u},{j}}\}}}
\quad&\mathop{\mathrm{min}}_{\forall m}\Bigg({\theta}_{1,m}\alpha_m^{{\theta}_{2,m}}
+{\theta}_{3,m}\Bigg) \nonumber
\\\mathrm{s.t.}\quad\quad&
\alpha_{m}\leq\sum_{u\in\mathcal{G}_{c,m}}\sum_{j=1}^Jt_{u,j}B/A
\nonumber\\
&\mathrm{log}_2\left(1+s^{\diamond}_j\times F_{{u},{j}}{p_{{u},{j}}}\right),  \quad\forall {c,m},\label{P3a}
\\
&\mathbf{E}\in \mathcal{Q}(\mathbf{s}^{\diamond}), \label{P3b}
\\
&(1-s^{\diamond}_{j})t_{{u},{j}}=0,\quad \forall {u,j}, \label{P3c}
\\
&t_{{u},{j}}\geq 0, \quad p_{{u},{j}}\geq 0,\quad \forall {u,j}, \label{P3d}
\\
&\sum_{j=1}^J\sum_{u=1}^Ut_{{u},{j}}+\frac{\mathrm{Tr}(\mathbf{D}^T\mathbf{E})}{v}\leqslant T_{all}, \label{P3e}
\\
&\sum_{j=1}^J\sum_{u=1}^Up_{{u},{j}}t_{{u},{j}}+\epsilon\Theta(\mathbf{E})\leqslant E_{all}. \label{P3f}
\end{align}
\end{subequations}

To solve $\mathrm{P}3$, the first challenge comes from the discrete constraint of $\mathbf{E}$ in \eqref{P3b}.
Fortunately, it can be seen that the variable $\mathbf{E}$ is only involved in \eqref{P3b}, \eqref{P3e} and \eqref{P3f}.
Furthermore, decreasing $\mathrm{Tr}(\mathbf{D}^T\mathbf{E})$ in \eqref{P3e} would also decrease $\Theta(\mathbf{E})$ in \eqref{P3f}, which enlarges the feasible set of $\mathrm{P}3$.
As a result, the optimal solution $\mathbf{E}^\diamond$ to $\mathrm{P}3$ is given by:
\begin{align}
\mathbf{E}^\diamond=\mathop{\mathrm{argmin}}\limits_{\mathbf{E}}~\{\mathrm{Tr}(\mathbf{D}^T\mathbf{E}):\mathbf{E}\in \mathcal{Q}(\mathbf{s}^{\diamond})\}. \label{tsp}
\end{align}

Since the above problem \eqref{tsp} is a travelling salesman problem, $\mathbf{E}^{\diamond}$ can be efficiently found by one-tree relaxation algorithm in the CVX Mosek solver\cite{mosek}. With $\mathbf{E}=\mathbf{E}^\diamond$, the optimization $\mathrm{P}3$ can be transformed into the following problem involving $\left\{t_{{u},{j}},p_{{u},{j}},\alpha_m\right\}$:
\begin{subequations}
\begin{align}
\mathrm{P}4:
\mathop{\mathrm{max}}\limits_{\substack{\left\{\alpha_m\right\},\\ \left\{t_{{u},{j}},p_{{u},{j}}\right\}}}
\quad&\mathop{\mathrm{min}}_{\forall m} \Bigg({\theta}_{1,m}\alpha_m^{{\theta}_{2,m}}
+{\theta}_{3,m}\Bigg) \nonumber
\\\mathrm{s.t.}\quad\quad&
\alpha_m\leq\sum_{u\in\mathcal{G}_{c,m}}\sum_{j=1}^J{t_{{u},{j}}}B\Big/A \nonumber\\
&\mathrm{log}_2\left(1+{{s}^{\diamond}_{j}} F_{{u},{j}}p_{{u},{j}}\right),\quad \forall {c,m}, \label{P4b}\\
&t_{{u},{j}}\geq 0, \quad p_{{u},{j}}\geq 0,\quad \forall {u,j}, \\
&\sum_{j=1}^J\sum_{u=1}^Ut_{{u},{j}}p_{{u},{j}}+\epsilon\Theta(\mathbf{E}^{\diamond})\leqslant E_{all}, \label{P4c}\\
&\sum_{j=1}^J\sum_{u=1}^Ut_{{u},{j}}+\frac{\mathrm{Tr}(\mathbf{D}^T\mathbf{E}^{\diamond})}{v}\leqslant T_{all}, \label{P1d}\\
&(1-s^{\diamond}_{j})t_{{u},{j}}=0,\quad \forall {u,j}.
\end{align}
\end{subequations}

However, the dimension of variables in $\mathrm{P}4$ is $M+2UJ$, which is too large if the number of vertices $J$ is large. To this end, the following proposition is established to reduce the problem dimension.
\begin{theorem}
The optimal $\{t^{\diamond}_{{u},{j}},p^{\diamond}_{{u},{j}}\}$ to \text{P4} satisfies:
\begin{itemize}
\item[(i)] If $t^{\diamond}_{{u},{j}}\neq 0$, then $p^{\diamond}_{{u},{j}}\neq 0$.
\item[(ii)]  If $j=\mathrm{argmax}_{l\in\mathcal{J}}{s}^{\diamond}_{l}F_{{u},{l}}$, then $t^{\diamond}_{{u},{j}}\neq 0$.
\item[(iii)] If $j\neq \mathrm{argmax}_{l\in\mathcal{J}}{s}^{\diamond}_{l}F_{{u},{l}}$, then $t^{\diamond}_{{u},{j}}=0$.
\end{itemize}
\end{theorem}

\begin{proof}
See Appendix B
\end{proof}

\textbf{Insights of Theorem 1.}
Part (i) of \textbf{Theorem 1} indicates that the UGV should leave no blank time, which is in contrast to the fixed EL systems that may have blank time for energy saving.
Part (ii) and (iii) of \textbf{Theorem 1} indicates that each IoT device should be served only once along the path, and the corresponding position has the best channel among all selected positions to that device.

According to \textbf{Theorem 1},
we can add the following equality constraints to $\mathrm{P}4$ without changing the optimal solution:
\begin{align}
&t_{u,j}=\left\{
\begin{aligned}
&e_u
,&\mathrm{if}~j=\mathop{\mathrm{argmax}}_{l\in\mathcal{J}}~{s}^{\diamond}_{l}F_{{u},{l}}&
\\
&0,&\mathrm{if}~j\neq\mathop{\mathrm{argmax}}_{l\in\mathcal{J}}~{s}^{\diamond}_{l}F_{{u},{l}}&
\end{aligned}
\right.,
\label{tuj}
\end{align}
and
\begin{align}
&p_{u,j}=\left\{
\begin{aligned}
&f_u
,&\mathrm{if}~j=\mathop{\mathrm{argmax}}_{l\in\mathcal{J}}~{s}^{\diamond}_{l}F_{{u},{l}}&
\\
&0,&\mathrm{if}~j\neq\mathop{\mathrm{argmax}}_{l\in\mathcal{J}}~{s}^{\diamond}_{l}F_{{u},{l}}&
\end{aligned}
\right.,
\label{puj}
\end{align}
where $e_{u}>0$ according to (ii) of \textbf{Theorem 1}, and $f_{u}>0$ according to (i) of \textbf{Theorem 1}.
\textbf{Theorem 1} indicates that for each $m$, the UGV only needs to allocate time to device $u$ at vertex $j=\mathrm{argmax}_{l\in\mathcal{J}}{s}^{\diamond}_{l}F_{{u},{l}}$, and should save transmission time and energy at other vertices.
Subsequently, by putting \eqref{tuj} and \eqref{puj} into $\mathrm{P}7$ (an equivalent form of $\mathrm{P}4$, see proof of theorem 1 in appendix B), $\mathrm{P}4$ is equivalently transformed into:
\begin{subequations}
\begin{align}
\mathrm{P}5:
\mathop{\mathrm{max}}\limits_{\substack{\left\{\alpha_m\right\},\\ \left\{e_u>0,f_u>0\right\}}}
\quad&\Bigg({\theta}_{1,m}\alpha_m^{{\theta}_{2,m}}
+{\theta}_{3,m}\Bigg) \nonumber
\\\mathrm{s.t.}\quad\quad&
\alpha_{m}-\sum_{u\in\mathcal{G}_{c,m}}{e_{u}}B\Big/A \nonumber\\
&\mathrm{log}_2\left(1+M_u(s^{\diamond})f_u\right)\leq 0,\quad \forall c,m \label{P6b}\\
&\sum_{u=1}^Ue_{u}f_u+\epsilon\Theta(\mathbf{E}^{\diamond})\nonumber\\
&\leqslant \beta_{E,m}E_{all}, \quad \forall {m}, \label{P6c}\\
&\sum_{u=1}^Ue_{u}+\frac{\mathrm{Tr}(\mathbf{D}^T\mathbf{E}^{\diamond})}{v}\nonumber\\
&\leqslant \beta_{T,m}T_{all}, \quad \forall {m},\label{P6d}\\
&e_u>0,\quad f_u>0,\\
&\sum_{m=1}^M\beta_{E,m}=1, \quad \sum_{m=1}^M\beta_{T,m}=1,
\end{align}
\end{subequations}
where $M_u(\mathbf{s}^{\diamond}):=\mathrm{max}_{l}s_{l}^{\diamond}F_{{u},{l}}$, which is a constant with $s_{l}=s_{l}^\diamond$. However, the $\left\{e_u,f_u\right\}$ optimization problem in the form of $\mathrm{P}5$ is still difficult to solve due to the nonlinear coupling between $e_{u}$ and $f_{u}$ in (24a) and (24b). To address this problem, we replace $\left\{f_{u}\right\}$ with new variables $\left\{\delta_{u}\right\}$ such that $\left\{\delta_{u}:=t_{u}f_{u}\right\}$. As such, the optimization problem $\mathrm{P}5$ becomes:
\begin{subequations}
\begin{align}
\mathrm{P}6:\quad\quad\quad\quad&
\nonumber\\
\mathop{\mathrm{max}}\limits_{\substack{\left\{\alpha_m\right\},\\ \left\{e_u>0,\delta_u>0\right\}}}
\quad&\Bigg({\theta}_{1,m}\alpha_m^{{\theta}_{2,m}}
+{\theta}_{3,m}\Bigg) \nonumber
\\\mathrm{s.t.}\quad\quad&
\alpha_{m}-\sum_{u\in\mathcal{G}_{c,m}}{e_{u}}B\Big/A \nonumber\\
&\mathrm{log}_2\left(1+\frac{\delta_u}{e_{u}} M_u(s^{\diamond})\right)\leq 0,\quad \forall c,m \label{P6b}\\
&\sum_{u=1}^U\delta_{u}+\epsilon\Theta(\mathbf{E}^{\diamond})\leqslant \beta_{E,m}E_{all}, \quad \forall {m}, \label{P6c}\\
&\sum_{u=1}^Ue_{u}+\frac{\mathrm{Tr}(\mathbf{D}^T\mathbf{E}^{\diamond})}{v}\nonumber\\
&\leqslant \beta_{T,m}T_{all}, \quad \forall {m},\label{P6d}\\
&e_u>0,\quad \delta_u>0,\\
&\sum_{m=1}^M\beta_{E,m}=1, \quad \sum_{m=1}^M\beta_{T,m}=1,
\end{align}
\end{subequations}
It can be seen that the variable dimension in problem $\mathrm{P}6$ is only $2U+M$, which is reduced by a factor of $J$ compared with the variable dimension in problem $\mathrm{P}4$. Therefore, solving $\mathrm{P}6$ requires a complexity of $\emph{O}\left(\left[2U+M\right]^{3.5}\right)$, which is significantly smaller than $\emph{O}\left(\left[2UJ+M\right]^{3.5}\right)$ for solving $\mathrm{P}4$.
Meanwhile, the solutions $\left\{e^{\diamond}_{u},\delta^{\diamond}_{u},\alpha^{\diamond}_{m}\right\}$ can be obtained by the following proposition.
\begin{proposition}
$\mathrm{P}6$ is a convex optimization problem.
\end{proposition}
\begin{proof}
See Appendix C
\end{proof}

According to \textbf{Proposition 1}, $\mathrm{P}6$ can be solved by the existing solvers, e.g. CVX \cite{cvx}. After obtaining the solutions $\left\{e^{\diamond}_{u},\delta^{\diamond}_{u} \right\}$, the optimal $\left\{f_{u}^{\diamond}\right\}$ to $\mathrm{P}6$ can be computed by
$\left\{f_{u}^{\diamond}={\delta_{u}^{\diamond}}/{e_{u}^{\diamond}}\right\}$.
\subsection{Practical Implementation}
The entire procedure of EL at UGV is shown in
Fig. 3. It can be seen from Fig. 3 that the EL at UGV is divided into three phases: 1) parameter estimation phase; 2) data collection with JPESP problem phase; and 3) model training at the edge phase.
The AlexNet is used in phase 1 and phase 3.
\begin{figure}[!t]
\centering
\includegraphics[width=85mm]{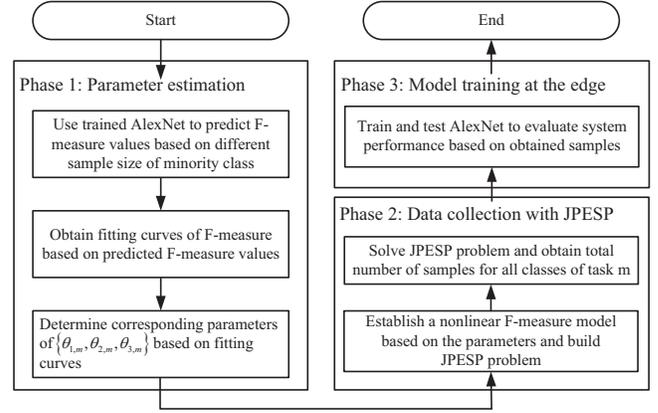}
\caption{The entire procedure for practical implementation.}
\label{fig_sim}
\end{figure}

During phase 1, one may wonder how could one obtain the fitted F-measure model before the
actual AlexNet model is being trained. This can be done via extrapolation \cite{obtaindata1}. More specifically,
in real-world applications, the data is being generated everyday at the IoT infrastructures \cite{obtaindata2}.
For example, in Internet of vehicles \cite{obtaindata3}, the on-board data (e.g., images, videos, LiDAR point
clouds) needs to be uploaded to the edge or cloud platform for big data analysis. In such a
scenario, it is not likely that the data is uploaded only once. Instead, the data is uploaded
multiple times \cite{obtaindata4} (e.g., twice a day/week). Therefore, there exist historical dataset samples at
the UGV edge and the F-measure can be obtained by fitting the corresponding coefficients
to the historical dataset. More specifically, by varying the samples sizes from $50$ to $1700$, we can obtain $(\theta_{1,1},\theta_{2,1},\theta_{3,1})=(-3.167,-0.3208,1.149)$ for task 1 (the purple dashed line in Fig. 2 (a)) and $(\theta_{1,2},\theta_{2,2},\theta_{3,2})=(-1.035,-0.4193,0.9456)$ for task 2 (the black dashed line in Fig. 2 (a)). It can be seen from Fig. 2 (a) that the fitting performance based on small historical dataset (i.e. the purple and black dashed lines) is slightly worse than that based on full dataset (i.e. the red and blue solid lines). But since our goal is to distinguish different task difficulties rather than accurate prediction of the F-measure values, the extrapolation method can still guide the JPESP design at the UGV and the IoT devices.

Notice that if there exist historical samples at the UGV, the total number of samples $y_{c,m}$ obtained for training the class $c$ of task $m$ (7) is redefined as:
\begin{align}\textcolor{red}{\label{26}}
y_{c,m}=\sum_{u\in\mathcal{G}_{c,m}}\sum_{j=1}^Jx_{u,j}+a_{c,m},
\end{align}
where $a_{c,m}$ is the number of historical samples for class $c$ of task $m$ at the UGV edge.

Accordingly, the first constraint of $\mathrm{P}1$ becomes:
\begin{align}
\alpha_{m}&\leq\sum_{u\in\mathcal{G}_{c,m}}\sum_{j=1}^Jt_{u,j}\frac{B}{A}\mathrm{log}_2\left(1+s_j\frac{p_{{u},{j}}\left|h_{{u},{j}}\right|^2}{{\sigma}^{2}}\right)\nonumber\\
&\quad{}{}
+\sum_{c=1}^{C_{m}}a_{c,m},\quad \forall {c,m}
\end{align}
comparing (27) with the first constraint in $\mathrm{P}1$, it can be seen that the
only difference is the additional constant term $\sum_{c=1}^{C_{m}}a_{c,m}$ in (27). Therefore, the convexity of
the problem is unchanged, and the proposed \textbf{Algorithm 1} is still applicable.

\vspace{0.2in}
\section{Simulation And Discussion}

In this section, simulations are provided to evaluate the system performance of the EL-UGV system. The simulation settings are set as follows unless specified otherwise. The time budget $T_{all}$ and energy budget $E_{all}$ are set to be $T_{all}=300$ s and $E_{all}=2000$ Joule, respectively. The data collection map is a $20\,\textrm{m} \times 20\,\textrm{m}=400\,\mathrm{m}^{2}$ square area. Within the map, $J=12$ vertices are uniformly scattered wherein $j=1$ is chosen as the starting point of the UGV. The channel $h_{{u},{j}}$ from IoT device $u$ to vertex $j$ is modeled as
$\sim \mathcal{CN}(0,\rho_{0}(d_{{u},{j}}/d_{0})^{-3})$ \cite{channel}, where $d_{{i},{j}}$ is the distance between $u$ to $j$ and $\rho_{0}=10^{-3}$ is the path-loss at $d_{0}=1$, and the channel bandwidth $B$ is set as $0.18$ MHz.
For comparison, two benchmark schemes are simulated: 1)\emph{ Fixed EL} (i.e., optimal solution to $\mathrm{P}1$ with $\mathbf{s}=[1,0,...,0]^{T}$; 2) \emph{full path EL} (i.e., optimal solution to $\mathrm{P}1$ with $\mathbf{s}=[1,1,...,1]^{T}$).

\subsection{Verification of Theoretical Results}

To verify the benefit brought by the F-measure model, we first simulate the \emph{fixed EL} and \emph{full path EL} schemes with or without the F-measure model.
For the schemes without the F-measure model, the power and time resources are optimized for maximizing the minimum throughput among all devices \cite{throughput1}.
Fig.~4 (a) shows the learning performance versus the noise power $\sigma^{2}$.
It can be seen that the \emph{full path EL} scheme achieves better performance that that of the \emph{fixed EL} scheme, and the gap quantifies the benefit by allowing the edge to move.
Furthermore, no matter the UGV is fixed at the starting point or moving around, the schemes based on the F-measure model always achieves much higher learning performance than those without F-measure model.
This indicates that it is necessary to adopt the F-measure model for energy and sample size planning .

\begin{figure}[t]
\center
\includegraphics[width=85mm]{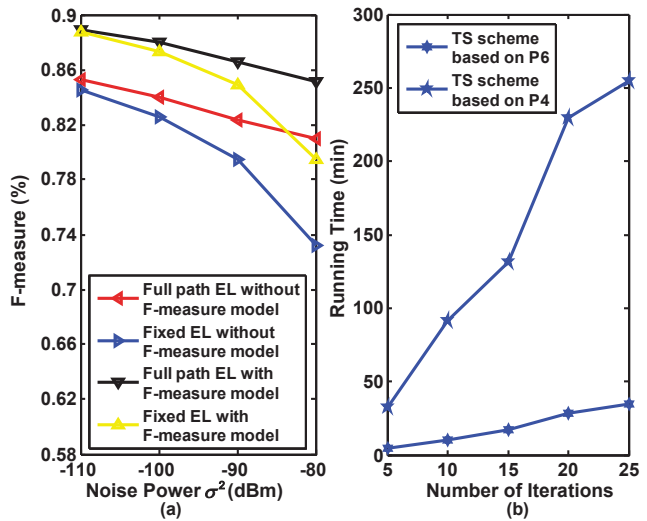}
\caption{Verification of theoretical results. (a) is the comparison between schemes with and without the F-measure model. (b) is running time versus the number of iterations for TS scheme based on $\mathrm{P}4$ and $\mathrm{P}6$.}
\label{fig:}
\end{figure}

Next, to verify the significance of \textbf{Theorem 1}, we consider the following schemes: 1) \emph{TS scheme based on $\mathrm{P}6$}; 2) \emph{TS scheme based on $\mathrm{P}4$} for running time comparison using Matlab on a computer equipped with Intel Core E5-26200 2GB, and 32GB RAM memory. Fig.~4 (b) shows the running time of the two schemes versus the number of iterations, with each point in Fig. 4 (b) averaged over 10 independent channel realizations.
From Fig. 4 (b), it is observed that TS scheme based on $\mathrm{P}6$ takes less running time than that based on $\mathrm{P}4$ at the same number of iterations.
On the other hand, since $\mathrm{P}4$ is equivalent to $\mathrm{P}6$, the TS scheme based on $\mathrm{P}4$ and $\mathrm{P}6$ would require the same number of iterations for convergence. Therefore, TS scheme based on $\mathrm{P}6$ is faster than that based on $\mathrm{P}4$. Moreover, the gap of running time between the two schemes significantly increases as the number of iterations increase.
This corroborates the fact that the variable dimension of $\mathrm{P}6$ is reduced by using \textbf{Theorem 1}.

\subsection{Single-Task Scenario}

Now, we consider a single-task scenario with total $U=10$ IoT devices, and classify 10 different categories using CIFAR-10 dataset.
Fig. 5 (a) shows the F-measure versus the noise power $\sigma^{2}$.
It can be observed from Fig. 5 (a) that the performance of all schemes decreases when the noise power increases, as a larger noise leads to a smaller data-rate, which in turn decreases the learning performance.
Furthermore, the proposed UGV-based TS algorithm achieves the best learning performance compared to the \emph{full path EL} scheme and the \emph{fixed EL} scheme.
More specifically, when $\sigma^{2}$ is small (e.g., -110 dBm), the UGV can easily collect data samples without moving to far-away vertices (with the moving path being the red dashed line shown in Fig. 5 (b)), which saves motion energy and time consumption compared to \emph{full path EL} scheme. Therefore, the proposed TS scheme collects more training samples and achieves a higher F-measure. On the other hand, when $\sigma^2$ is large (e.g., -80 dBm), the proposed scheme allows the UGV to move closer to far-away IoT devices (with the moving path being the black line shown in Fig. 5 (b)), which combats against the noisy channel by reducing the path-loss. Therefore, the proposed TS scheme also collects more training samples and outperforms the \emph{fixed EL} scheme.

To further evaluate the performance of the proposed algorithm, practical items including the the sample size, F-measure, the total transmission time, and the overall transmission energy of the three schemes are shown in Fig.~5 (c). It can be seen from Fig. 5 (c) that the proposed scheme achieves the highest F-measure score (0.865).
This is because the proposed algorithm can automatically determine whether to move and how far to move, and find the best trade-off between moving and communicating.
Therefore, with the proposed algorithm, the edge collects the most samples (2260) for the minority class.
\begin{figure*}
\centering
\subfigure[F-measure versus noise power $\sigma^{2}$.]{\label{fig:subfig:a}
\includegraphics[width=0.3\linewidth]{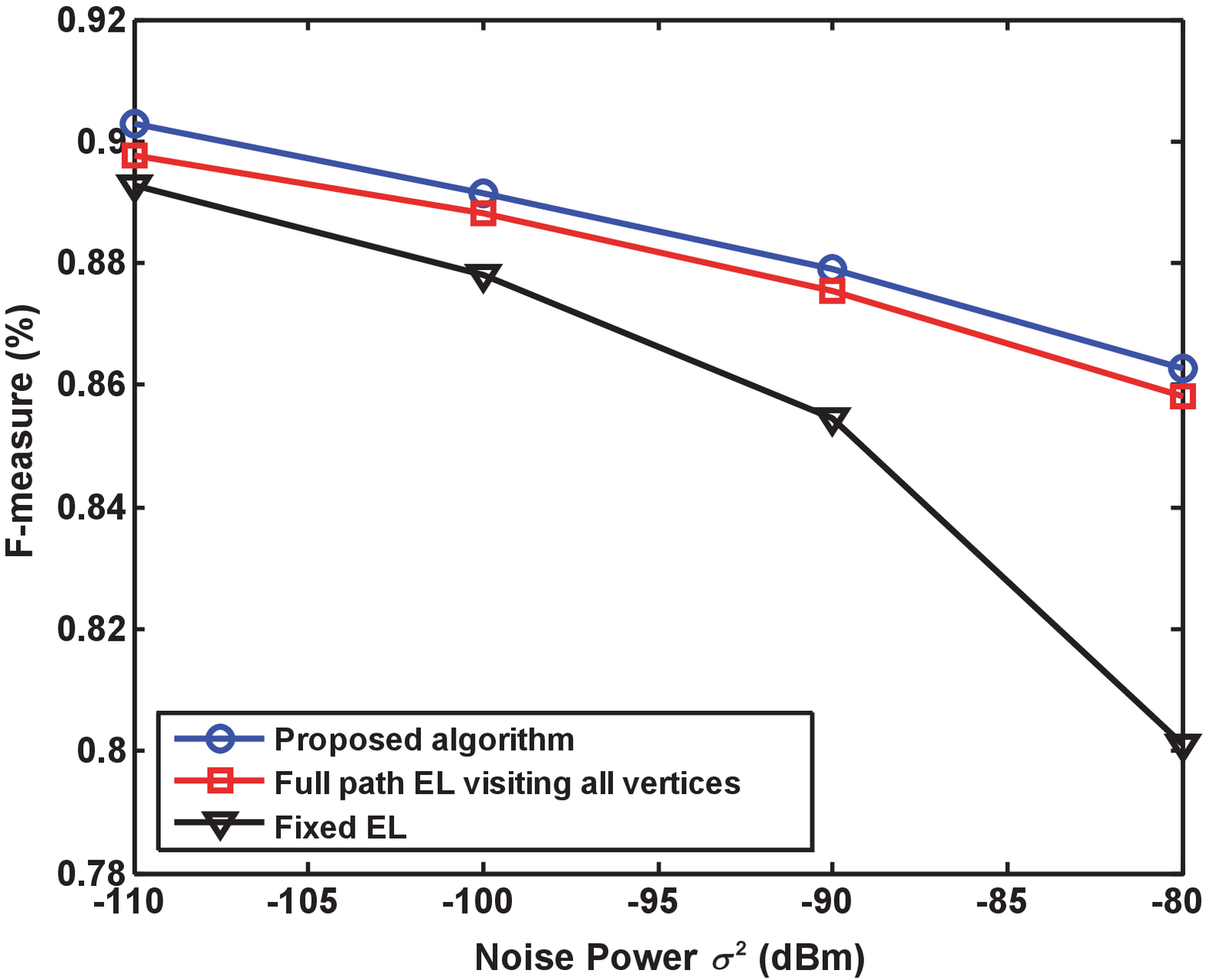}}
\hspace{0.01\linewidth}
\subfigure[The optimal path at $\sigma^{2}=-80~\rm{dBm}$ and $\sigma^{2}=-110~\rm{dBm}$.]{\label{fig:subfig:b}
\includegraphics[width=0.3\linewidth]{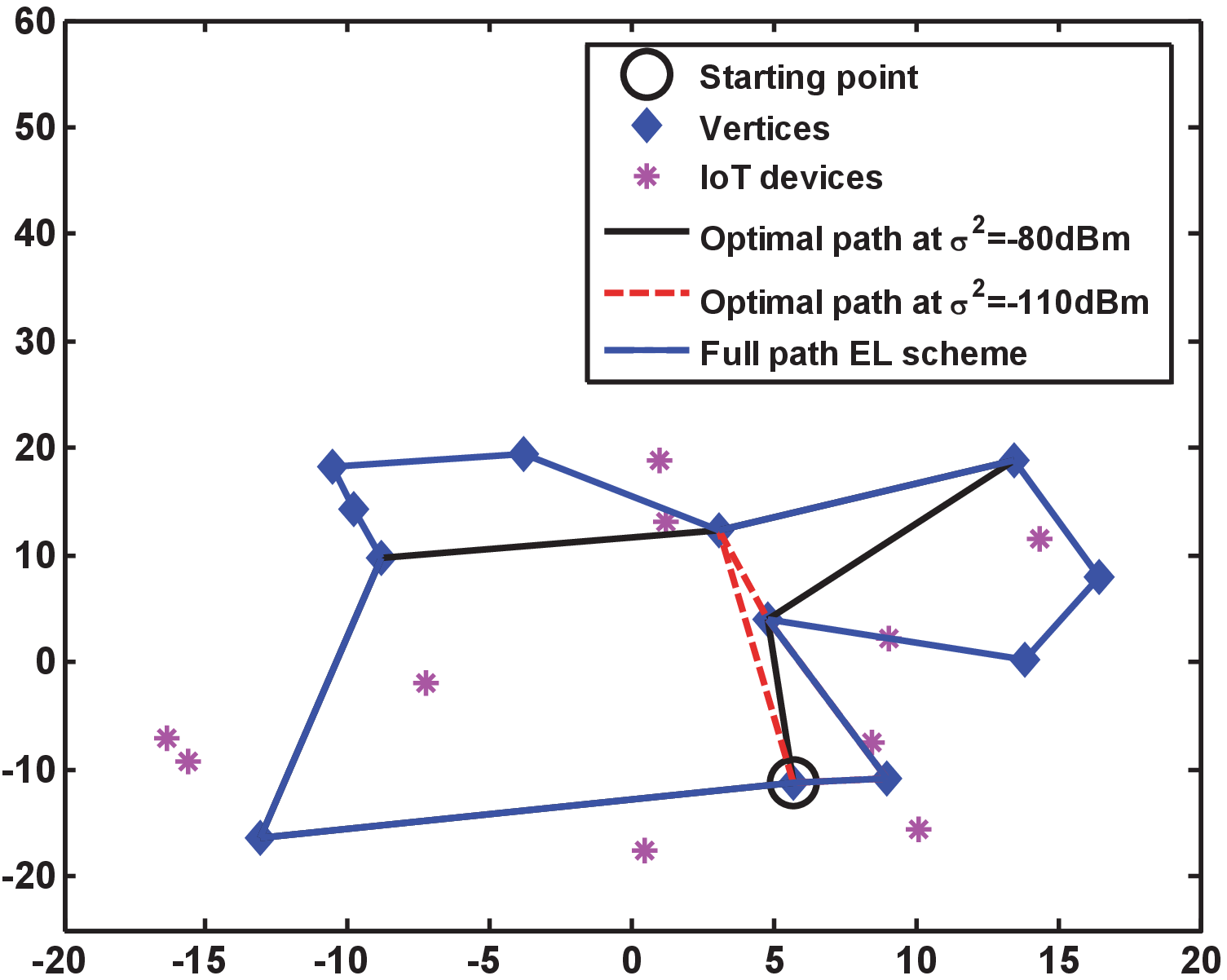}}
\hspace{0.01\linewidth}
\subfigure[Comparison of sample size, F-measure, total transmission time, overall transmission energy at $\sigma^{2}=-80~\rm{dBm}$.]{\label{fig:subfig:c}
\includegraphics[width=0.3\linewidth]{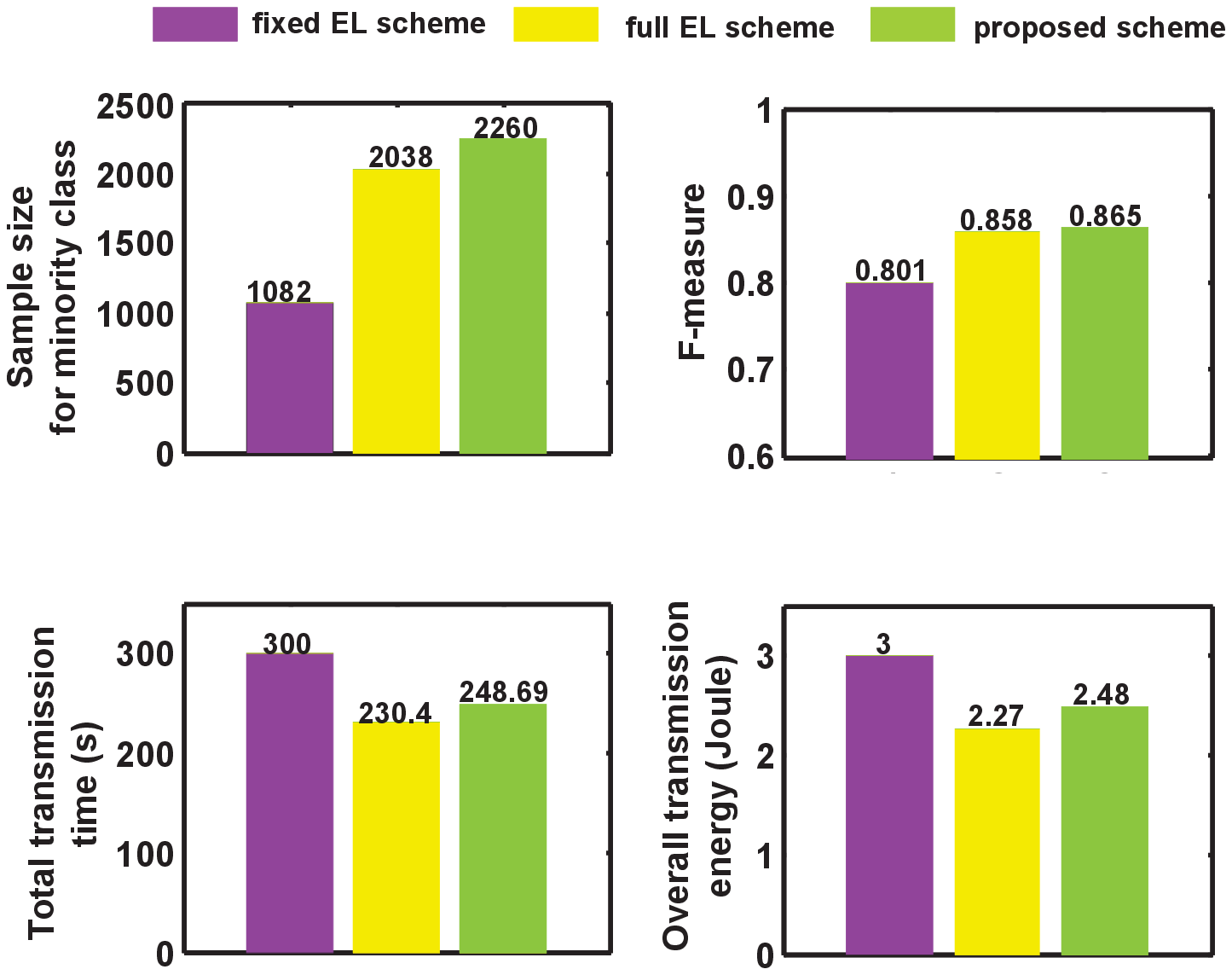}}
\caption{The system performance and optimal path for single-task scenario wherein $U=10, J=12$.}
\label{fig:subfig}
\end{figure*}

\subsection{Two-Task Scenario}
\begin{figure*}
\centering
\subfigure[Total F-measure versus noise power $\sigma^{2}$.]{\label{fig:subfig:a}
\includegraphics[width=0.3\linewidth]{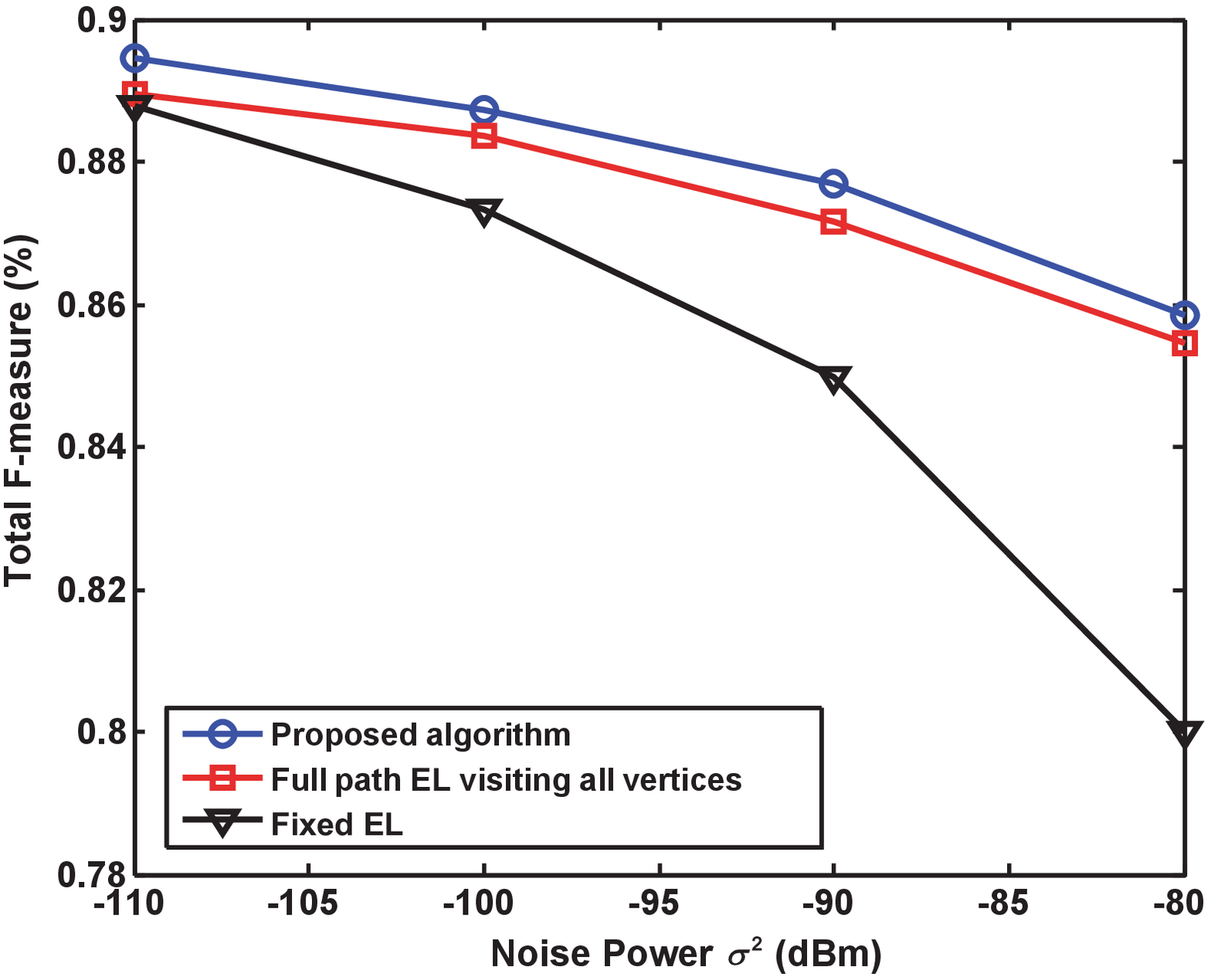}}
\hspace{0.01\linewidth}
\subfigure[The optimal path at $\sigma^{2}=-80~\rm{dBm}$ and $\sigma^{2}=-110~\rm{dBm}$.]{\label{fig:subfig:b}
\includegraphics[width=0.3\linewidth]{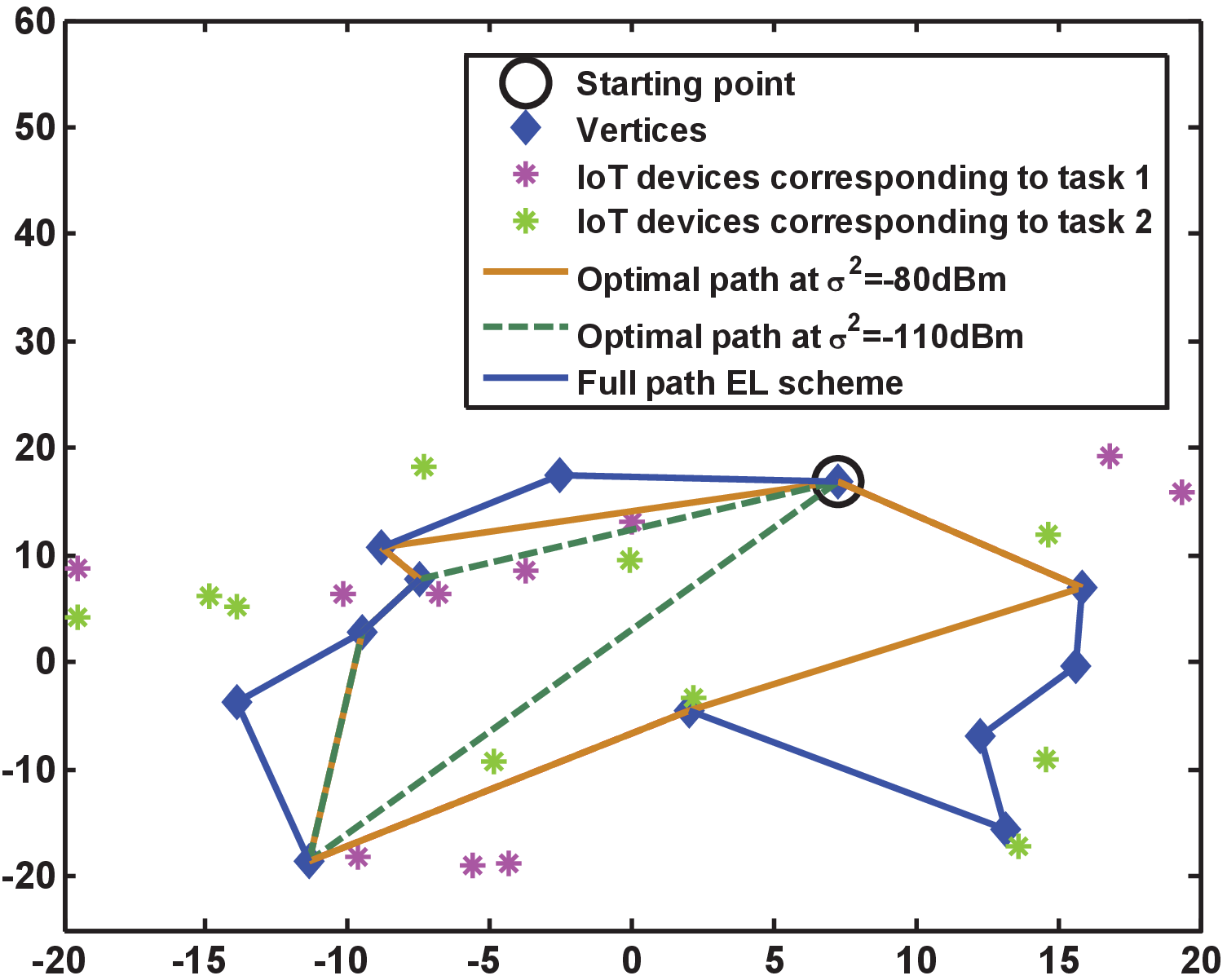}}
\hspace{0.01\linewidth}
\subfigure[Comparison of sample size, total F-measure, total transmission time, overall transmission energy at $\sigma^{2}=-80~\rm{dBm}$.]{\label{fig:subfig:c}
\includegraphics[width=0.3\linewidth]{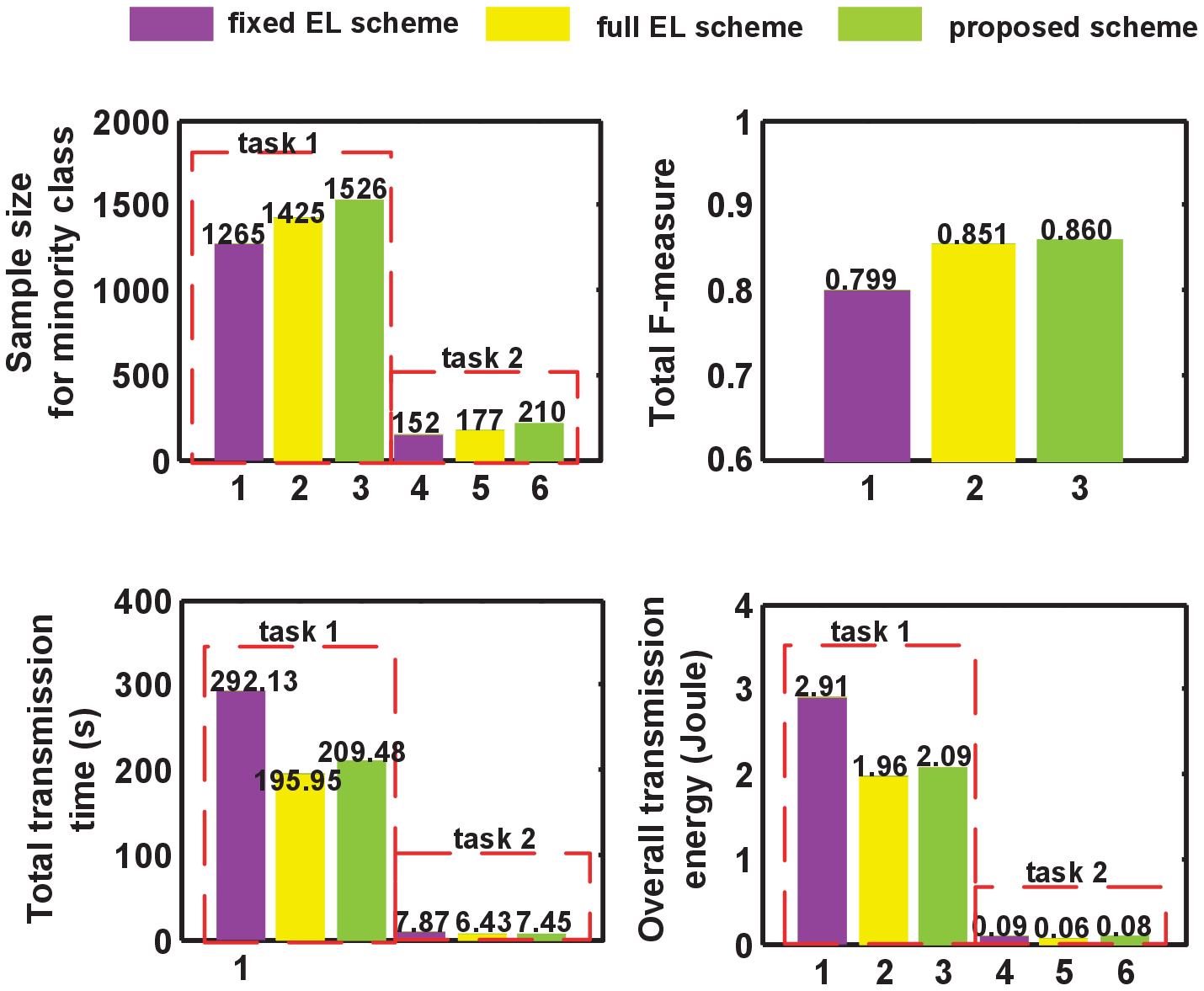}}
\caption{The system performance and optimal path for two-task scenario when $U=20, J=12$.}
\label{fig:subfig}
\end{figure*}

In such scenario, we consider the case of $U=20$ and $M=2$. The first task is to classify the CIFAR-10 dataset into 10 categories, and the second task is to classify the MNIST fashion dataset into 10 categories. The 20 categories are distributed at the data of these 20 devices, with each device providing the training samples of a particular category.
We adopt the values of $\{\theta_{1,m}, \theta_{2,m}, \theta_{3,m}\}_{m=1}^M$ in Section IV-B.

To evaluate the superiority of the proposed algorithm, we compare the total F-measure (i.e. we first compute F-measure of task 1 and F-measure of task 2 respectively, and then take the minimum of the two F-measures as the total F-measure) of different schemes. The performance brought by total F-measure model is demonstrated in Figure 6 (a). It can be seen from Fig. 6 (a) that the proposed algorithm still achieves the largest F-measure over a wide range of noise power. Similar to the single-task case, the UGV should increase its path length when the noise power increases as shown in Fig. 6 (b). However, in addition to the above feature, the proposed TS scheme always chooses the positions that are closer to the IoT devices of task 1 (as shown in Fig. 6 (b)). This is because the task 1 needs to classify the CIFAR-10 dataset, which is much more challenging than the MNIST dataset in task 2. As a result, in order to achieves the best F-measure learning performance, more time and energy resources need to be allocated to task 1 (the total transmission time is 209.48 s and the overall transmission energy is 2.09 Joule shown in Fig. 6 (c)) to collect more samples of it (1526). This indicates that the proposed TS scheme can not only find the best trade-off between moving and communicating, but also can adapt to the difficulty and importance of different tasks. Furthermore, it can be seen from Fig. 6 (c), for the two-task scenario, the proposed scheme achieves the highest F-measure of 0.86, and the sample sizes of it for both task 1 and task 2 are also the largest (1526 for task 1 and 210 for task 2).

\subsection{Comparison with Related Works}
Existing works have studied the energy planning problem \cite{compare1} or the UGV path planning problem \cite{compare2}.
While \cite{compare1},\cite{compare2} aim to maximize the communication throughput instead of the learning performance (e.g., F-measure adopted in this paper), they have employed the same multiple access scheme (i.e., time division multiple access (TDMA)) as our EL system. Therefore, the two algorithms can be applied to our considered UGV edge learning system with some minor modifications. Specifically,
\begin{itemize}
\item
For problem (4a)-(4d) in \cite{compare1}, we set $\tau_0 = 0$ in (4a)-(4d) and $\gamma_k=p_{u,1}|h_{u,1}|^2\sigma^2$ in (4a), and replace the constraint (4b) with $\sum_{k=1}^Kp_kt_k \leq E_{all}$. Then the TDMA scheme in \cite{compare1} becomes the throughput maximization scheme with fixed UGV in our EL systems.
\item
For problem (7) in \cite{compare2}, we set $q^a_{i}=x^a_{i}/log_2(x^a_{i})$ in (7) and write the data amount xai
based on the Shannon information theory as $x^2_{i}=t^a_{i}Blog_2(1+s^a)p^2_{i}|h^2_{i}|^2/\sigma^2$ (notice that
$h^a_{i}=h_{u,j}$). Then the problem (7) in \cite{compare2} becomes the throughput maximization scheme with UGV path planning.
\end{itemize}}

\begin{figure}\setcounter{subfigure}{0}
\centering
\subfigure[Comparison of total F-measure, total transmission time, overall transmission energy, and sample size at $\sigma^{2}=-80~\rm{dBm}$.] { \label{fig:a}
\includegraphics[width=0.9\columnwidth]{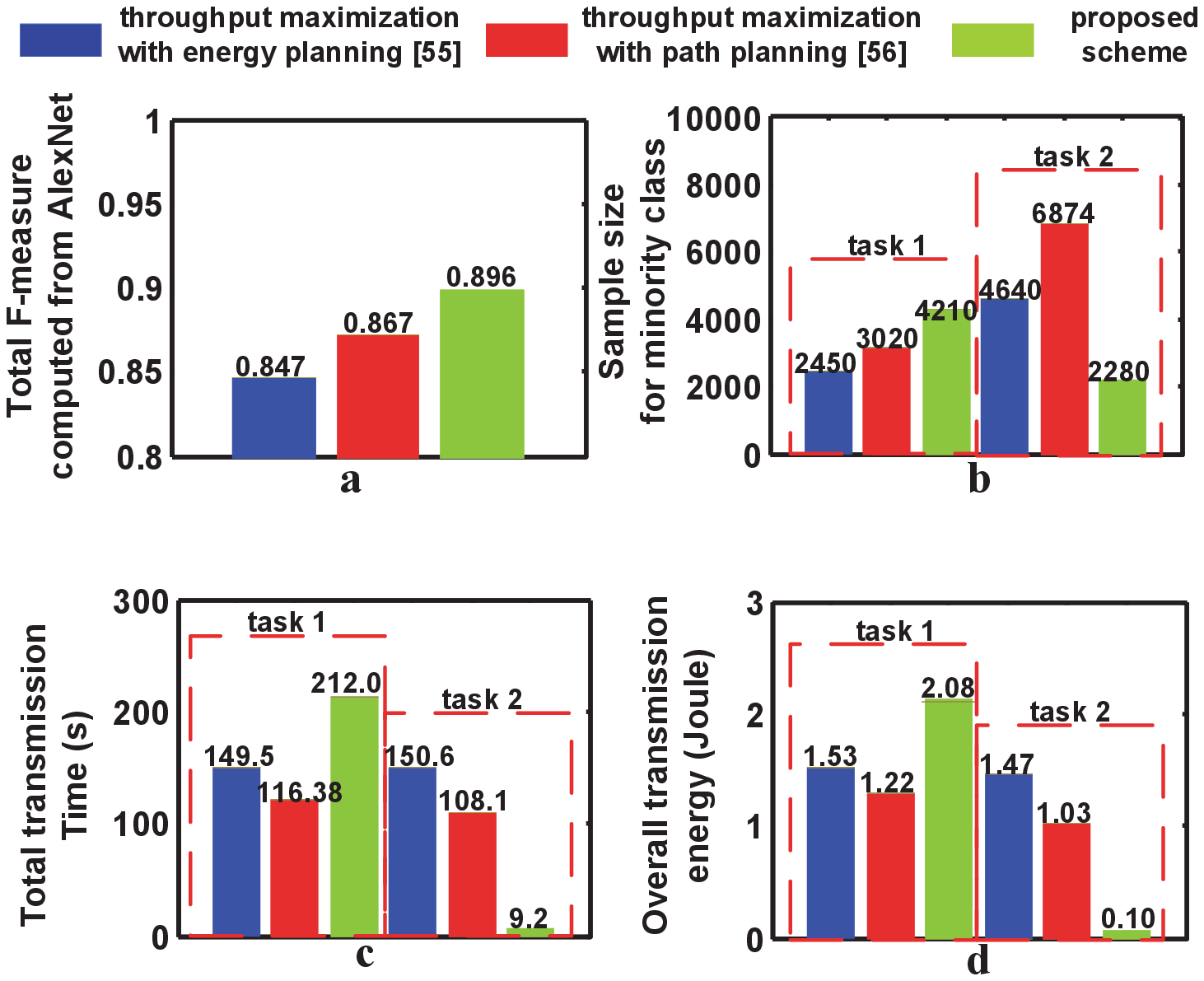}
}
\subfigure[The optimal path at $\sigma^2=-80 dBm$.] { \label{fig:b}
\includegraphics[width=0.9\columnwidth]{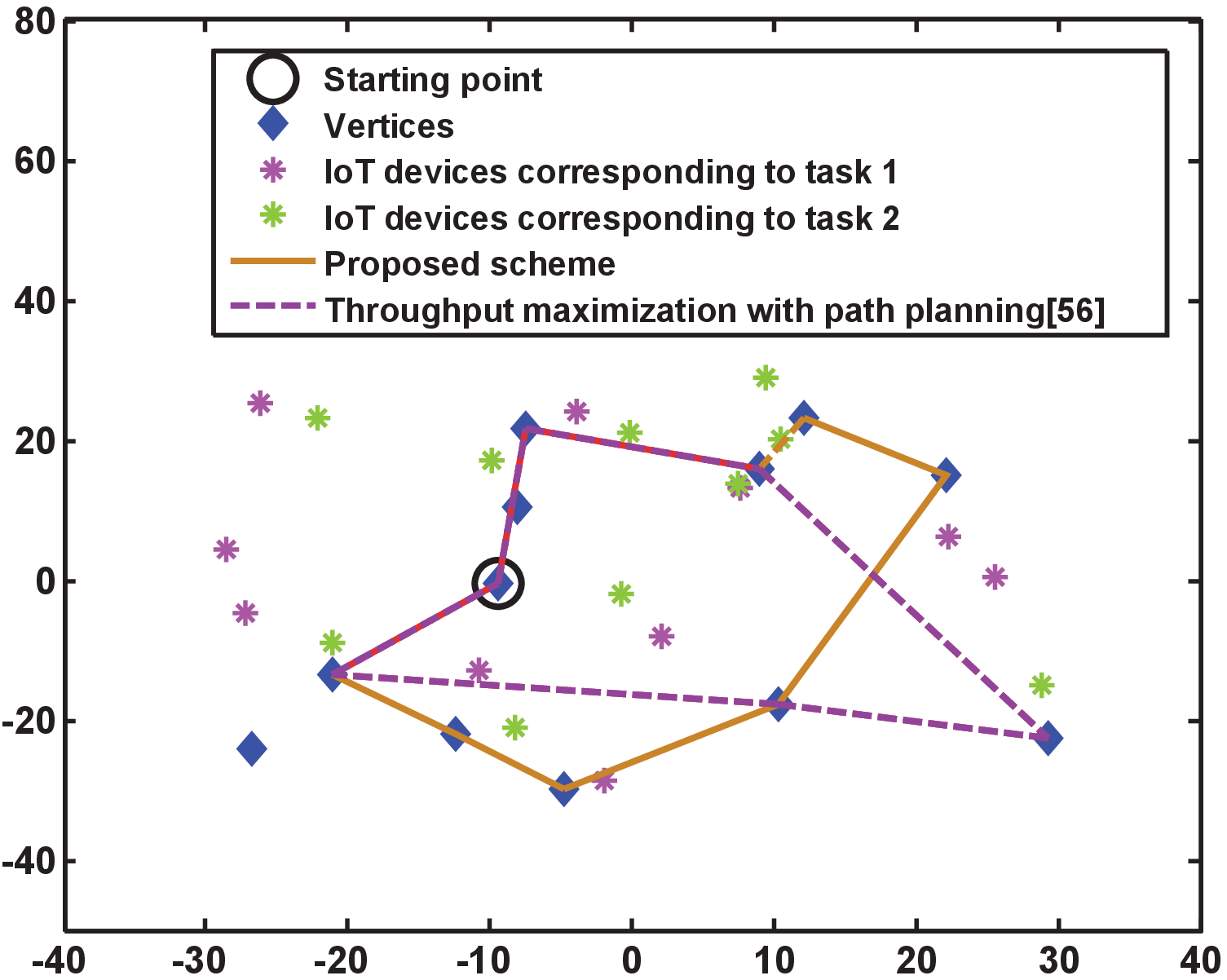}
}
\caption{Comparison with related works.}
\label{fig}
\end{figure}

Based on the above descriptions, we simulate the two-task scenario and compare our proposed EL-UGV scheme with the throughput maximization scheme with energy planning \cite{compare1} and the throughput maximization scheme with UGV path planning \cite{compare2}. For the communication procedure, we adopt the same parameters as those in Section VI but with the bandwidth $B=0.3$ MHz. For the learning procedure, the AlexNet is trained with a batch size of 32, a weight
decay of 0.0005, a momentum rate of 0.9, and a learning rate of 0.005. After training for
10000 iterations, the trained AlexNet is tested on a dataset with 1000 unseen images (each
class has 100 images). The parameters of the F-measure model are estimated by using 1700
historical samples (the number of historical samples of all classes is 1700), and they are given by $(\theta_{1,1},\theta_{2,1},\theta_{3,1})=(-1.961,-0.09712,1.795)$ for task 1
$(\theta_{1,2},\theta_{2,2},\theta_{3,2})=(-0.755,-0.08913,1.302)$ for task 2.

Under the above setting, the total transmission time, the total transmission energy and the minimum sample size among all classes are obtained by executing \textbf{Algorithm 1}.
For both deep learning tasks, we train the AlexNet according to the sample sizes of all classes obtained from the simulated algorithms (including \textbf{Algorithm 1}, [55] and [56]) and calculate the total F-measure scores for all the schemes.
It can be seen from Fig.~7 (a) that the proposed EL-UGV scheme achieves the highest F-measure score.
This is because the proposed scheme ``learns'' that the task 1 is more difficult than task 2, and allocates more transmission time and transmission energy to task 1.
In particular, with our proposed scheme, the transmission time for task 1 is 212.0 s as in Fig.~7 (a)c and the transmission energy for task 1 is 2.08 Joule as in Fig.~7 (a)d.
Both values are the largest among all the simulated schemes.
As a result, with the proposed scheme, the edge is able to collect more samples for training task 1 (collecting 4210 samples as shown in Fig.~7 (a)b), and therefore the F-measure is significantly increased compared with [55] and [56].
Notice that the proposed scheme achieves the smallest total transmission time (221.2 s) and total transmission energy (2.18 Joule), meaning that the proposed scheme can also benefit the energy saving at IoT devices.

\section{Conclusions}

In this paper, UGV was considered in EL for collecting data and training the deep learning models. In order to jointly plan the UGV path, the devices' energy consumption, and the number of training samples for different jobs, the graph-based path planning model, the network energy consumption model and the sample size planning model that is based on F-measure and minority class sample size were proposed. With these models, a JPESP problem that maximizes the minimum F-measure for all tasks was established. Since the problem is a large-scale MINLP, it was proved that each IoT device should be served only once along the path and a TS-based algorithm for optimizing the UGV path was derived. Simulations demonstrated that the proposed algorithm realizes a trade-off between moving and communicating, and adapts to the difficulty and importance for various tasks.

\appendices

\section{Confusion Matrix}
\begin{table}[!h]
\caption{{\bf Confusion Matrix}}
\centering
\begin{tabular}{lll}
\hline
Truth / Prediction & Positive  & Negative \\ \hline
Positive & $\mathrm{TP}_m$ & $\mathrm{FN}_m$  \\
Negative & $\mathrm{FP}_m$ & $\mathrm{TN}_m$ \\ \hline
\end{tabular}
\label{table1}
\end{table}

In Table~II, ``positive'' means ``minority class'' and ``negative'' means ``other classes''.
The values $(\mathrm{TP}_m,\mathrm{FN}_m,\mathrm{FP}_m,\mathrm{TN}_m)$ represent the number of samples under different combinations for task $m$.
For example, $\mathrm{TP}_m$ is the number of samples which belongs to the minority class while being classified correctly;
$\mathrm{FN}_m$ is the number of samples which belongs to the minority class while being classified incorrectly;
$\mathrm{TN}_m$ is the number of samples which belongs to other classes while being classified correctly.

\section{Proof of Theorem 1}
We assume that $\alpha_{c,m}=0$. To prove this theorem, we need the following lemma.
\begin{lemma}
There always exist $\{\beta_{E,m}, \beta_{T,m}\}$ with $\sum_{m=1}^M\beta_{E,m}=1$ and $\sum_{m=1}^M\beta_{T,m}=1$ such that $\mathrm{P}4$ is equivalent to :
\begin{subequations}
\begin{align}
\mathrm{P}7:
\mathop{\mathrm{max}}\limits_{\substack{\left\{\alpha_m\right\},\\ \left\{t_{{u},{j}},p_{{u},{j}}\right\}}}
\quad&\Bigg({\theta}_{1,m}\alpha_m^{{\theta}_{2,m}}
+{\theta}_{3,m}\Bigg) \nonumber
\\\mathrm{s.t.}\quad\quad&
\alpha_m\leq\sum_{u\in\mathcal{G}_{c,m}}\sum_{j=1}^J{t_{{u},{j}}}B\Big/A \nonumber\\
&\mathrm{log}_2\left(1+{{s}^{\diamond}_{j}} F_{{u},{j}}p_{{u},{j}}\right),\quad \forall {c,m}, \label{P4b}\\
&t_{{u},{j}}\geq 0, \quad p_{{u},{j}}\geq 0,\quad \forall {u,j}, \\
&\sum_{j=1}^J\sum_{u=1}^Ut_{{u},{j}}p_{{u},{j}}+\epsilon\Theta(\mathbf{E}^{\diamond})\nonumber\\
&\leqslant \beta_{E,m}E_{all}, \quad \forall {m},\label{P4c}\\
&\sum_{j=1}^J\sum_{u=1}^Ut_{{u},{j}}+\frac{\mathrm{Tr}(\mathbf{D}^T\mathbf{E}^{\diamond})}{v}\nonumber\\
&\leqslant \beta_{T,m}T_{all}, \quad \forall {m},\label{P1d}\\
&(1-s^{\diamond}_{j})t_{{u},{j}}=0,\quad \forall {u,j},
\end{align}
\end{subequations}
where $\beta_{E,m}$ denotes the proportion of energy requirement for completing task $m$ to $E_{all}$, and $\beta_{T,m}$ denotes the proportion of time requirement for completing task $m$ to $T_{all}$.
\end{lemma}
\begin{proof}
Assume an optimal solution $\left\{t^{\diamond}_{u,j},p^{\diamond}_{u,j}\right\}$ to $\mathrm{P}7$ with a particular $(a,b)$ such that $t^{\diamond}_{a,b}=t^{\vartriangle}$ and $p^{\diamond}_{a,b}=p^{\vartriangle}$. As $(t_{a,b},p_{a,b})=(t^{\vartriangle},p^{\vartriangle})$ is optimal to $\mathrm{P}7$, it must satisfy the constraints (28c) and (28d) of $\mathrm{P}7$, that is:
\begin{align}
&t^{\vartriangle}p^{\vartriangle}+C_1\leq\beta_{E,m}E_{all}, \quad \forall m, \\ \nonumber
&t^{\vartriangle}+C_2\leq\beta_{T,m}T_{all}, \quad \forall m,
\end{align}
where
\begin{align}
C_1=&\mathop{\sum\sum}\limits_{(u,j)\neq (a,b)}t_{u,j}p_{u,j}+\epsilon\Theta(\mathbf{E}^{\diamond}),\\ \nonumber
C_2=&\mathop{\sum\sum}\limits_{(u,j)\neq (a,b)}t_{u,j}+\frac{\mathrm{Tr}(\mathbf{D}^T\mathbf{E}^{\diamond})}{v}.\nonumber
\end{align}

Since $0<\beta_{E,m}<1$ and $0<\beta_{T,m}<1$ for any $m$, we have:
\begin{align}
&t^{\vartriangle}p^{\vartriangle}+C_1\leq E_{all}, \quad \forall m, \\ \nonumber
&t^{\vartriangle}+C_2\leq T_{all}, \quad \forall m.
\end{align}
That means the constraints (21c) and (21d) in $\mathrm{P}4$ are hold. Moreover, assuming the optimal solutions $\left\{\alpha^{\diamond}_{m}\right\}$ to $\mathrm{P}7$ such that $\left\{\alpha^{\diamond}_{m}\right\}=\left\{\alpha_m^{\vartriangle}\right\}$, then we have
\begin{align}
{\theta}_{1,m}(\alpha_m^{\vartriangle})^{{\theta}_{2,m}}+{\theta}_{3,m}\geq {\theta}_{1,m}\alpha_m^{{\theta}_{2,m}}+{\theta}_{3,m}, \quad \forall m,
\end{align}
which leads to
\begin{align}
\mathop{\mathrm{min}}_{ m}\Big({\theta}_{1,m}(\alpha_m^{\vartriangle})^{{\theta}_{2,m}}+{\theta}_{3,m}\Big)\geq \mathop{\mathrm{min}}_{m}\Big({\theta}_{1,m}\alpha_m^{{\theta}_{2,m}}+{\theta}_{3,m}\Big).
\end{align}
That means the objective function of $\mathrm{P}4$ holds. Meanwhile, as the constraints (21a) and (21e) of $\mathrm{P}4$ are the same as the constrains (28a) and (28e) of $\mathrm{P}7$, they always hold with the optimal solutions $\left\{\alpha_m^{\diamond},t_{a,b}^{\diamond},p_{a,b}^{\diamond}\right\}$ of $\mathrm{P}7$. In summary, the optimal solutions $\left\{\alpha_m^{\diamond},t_{u,j}^{\diamond},p_{u,j}^{\diamond}\right\}$ of $\mathrm{P}7$ are also optimal to $\mathrm{P}4$, which means $\mathrm{P}4$ can be equivalently converted into $\mathrm{P}7$.
\end{proof}

\subsection {Proof of Part (i)}
Then, we can prove this theorem based on $\mathrm{P}7$. We first prove the first part. Since the optimal solution $\left\{\alpha_{m}\right\}$ must activate the constraint (28a) of $\mathrm{P}7$, which leads to:
\begin{align}
\alpha_{m}=\sum_{u\in\mathcal{G}_{c,m}}\sum_{j=1}^J{t_{{u},{j}}}B\Big/A\mathrm{log}_2\left(1+s^{\diamond}_jF_{{u},{j}}{p_{{u},{j}}}\right), \quad \forall c,m
\end{align}
then $\mathrm{P}7$ can be re-expressed as:
\begin{subequations}
\begin{align}
\mathrm{P}8:
\mathop{\mathrm{max}}\limits_{\substack{\{\alpha_m\}\\\{t_{{u},{j}},p_{{u},{j}}\}}}
\quad&\Big({\theta}_{3,m}+{\theta}_{1,m}\Big(\sum_{u\in\mathcal{G}_{c,m}}\sum_{j=1}^Jt_{u,j}B/A
\nonumber\\
&\mathrm{log}_2\Big(1+s^{\diamond}_j F_{{u},{j}}{p_{{u},{j}}}\Big)\Big)^{{\theta}_{2,m}}\Big),\nonumber\\
\mathrm{s.t.}\quad\quad&(28b)-(28e).
\end{align}
\end{subequations}

Assume an optimal solution $\{t^{\diamond}_{{u}, {j}}, p^{\diamond}_{{u},{j}}\}$ to $\mathrm{P}8$ with a particular $(v,w)$ such that $t^{\diamond}_{{v}, {w}}=\hat{t}\neq 0$ and the corresponding $p^{\diamond}_{{v}, {w}}=0$, where user $v$ stores samples of class $c$ of task $m$. Consider the following related problem by fixing all the variables to their optimal values except for $(t_{{v}, {w}}, p_{{v}, {w}})$:
\begin{subequations}
\begin{align}
\mathrm{P}9:
\mathop{\mathrm{max}}\limits_{t_{{v}, {w}},p_{{v}, {w}}\geq 0}
\quad&\Big({\theta}_{1,m}\Big(A_1+t_{{v}, {w}}B/A\times \mathrm{log}_2\Big(1\nonumber\\
&+s^{\diamond}_wF_{{v},{w}}{p_{{v},{w}}}\Big)\Big)^{{\theta}_{2,m}}+{\theta}_{3,m}\Big),\nonumber\\
\mathrm{s.t.}\quad\quad&A_{2}+t_{{v},{w}}p_{{v},{w}}\leq \beta_{E,m}E_{all}, \quad \forall m\\
&A_{3}+t_{{v},{w}}\leq \beta_{T,m}T_{all}, \quad \forall m\\
&\sum_{m=1}^M\beta_{E,m}=1, \quad \sum_{m=1}^M\beta_{T,m}=1,
\end{align}
\end{subequations}
where
\begin{align}
A_{1}=&\mathop{\sum\sum}\limits_{(u,j)\neq (v,w)} t^{\diamond}_{{u},{j}}B/A\times \mathrm{log}_2\left(1+s^{\diamond}_j F_{{u},{j}}p^{\diamond}_{{u},{j}}\right),\nonumber\\
A_{2}=&\mathop{\sum\sum}\limits_{(u,j)\neq (v,w)}t^{\diamond}_{{u},{j}}p^{\diamond}_{{u},{j}}+\epsilon\Theta(\mathbf{E}^{\diamond}),\nonumber\\
A_{3}=&\mathop{\sum\sum}\limits_{(u,j)\neq (v,w)} t^{\diamond}_{{u},{j}}+{Tr(\mathbf{D}^T\mathbf{E}^{\diamond})}/{v}.\nonumber
\end{align}

As $\{t^{\diamond}_{{u}, {j}}, p^{\diamond}_{{u},{j}}\}$ is optimal to $\mathrm{P}8$, it can be shown that $(t_{{v}, {w}},p_{{v}, {w}})=(\hat{t},0)$ is optimal to $\mathrm{P}9$. Then, putting $(\hat{t},0)$ into the objective function of $\mathrm{P}9$, we have ${\theta}_{1,m}{A_{1}}^{{\theta}_{2,m}}+{\theta}_{3,m}$.
Considering another solution $(t_{{v}, {w}},p_{{v}, {w}})=(\tilde{t},\tilde{p})$ and putting it into the same objective function, we have
${{\theta}_{1,m}}\left(A_{1}+\tilde{t}B/A\times\mathrm{log}_2\left(1+s^{\diamond}_wF_{{v},{w}}\tilde{p}
\right)\right)^{{\theta}_{2,m}}+{{\theta}_{3,m}}$.
Due to $(t_{{v}, {w}}, \delta_{{v}, {w}})=(\hat{t},0)$ is of optimal to $\mathrm{P}9$, we can obtain the following inequality:
\begin{align}
A_{1}+\tilde{t}B/A\times \mathrm{log}_2\left(1+s^{\diamond}_{w}F_{{v},{w}}\tilde{p}\right)<A_{1}.
\end{align}

But since $s^{\diamond}_wF_{{v},{w}}\neq 0$ (as $t^{\diamond}_{{v},{w}}\neq 0$, $s^{\diamond}_w=1$ from (28e)), we have $\tilde{t}\times \mathrm{log}_2\left(1+s^{\diamond}_{w}F_{{v},{w}}\tilde{p}\right)>0$, which leads to $A_{1}+\tilde{t}B/A\times \mathrm{log}_2\left(1+s^{\diamond}_{w}F_{{v},{w}}\tilde{p}\right)>A_{1}$, and it contradicts to (37), that is to say $(t_{{v}, {w}},p_{{v}, {w}})=(\hat{t},0)$ cannot be optimal to $\mathrm{P}8$. Therefore, $p^{\diamond}_{{v}, {w}}\neq0$. Next, we prove the second part.

\subsection {Proof of Part (ii)}
To prove (ii), we also assume $\{t^{\diamond}_{{u}, {j}}, p^{\diamond}_{{u},{j}}\}$ as an optimal solution to $\mathrm{P}8$. Suppose that there exists user $v$ stored samples of class $c$ for task $m$ such that $t^{\diamond}_{{v},{w}}=0$ at vertex $w=argmax_{l\in\mathcal{J}}{s}^{\diamond}_{l}F_{{u},{l}}$, then there must exist $w^{'}\neq w$ such that $t^{\diamond}_{{v},{w'}}=\tilde{t}\neq 0$ and $p^{\diamond}_{{v},{w'}}=\tilde{p}\neq 0$.

Construct a new related problem of $\mathrm{P}8$ by fixing all the variables to their optimal values except for $(t_{{v},{w}},p_{{v},{w}},t_{{v},{w'}},p_{{v},{w'}})$:
\begin{subequations}
\begin{align}
\mathop{\mathrm{max}}\limits_{\substack{t_{{v},{w}},p_{{v},{w}}\geq 0 \\ t_{{v},{w^{'}}},p_{{v},{w^{'}}}\geq 0}}
\quad
&\Big({{\theta}_{1,m}}\Big(A_{11}+t_{{v},{w}}B/A\times  \mathrm{log}_2\Big(1 \nonumber\\
&+s^{\diamond}_wF_{{v},{w}}p_{{v},{w}}\Big)+t_{{v},{w'}}B/A\nonumber\\
&\times \mathrm{log}_2\Big(1+s^{\diamond}_{w^{'}}F_{{v},{w'}}p_{{v},{w'}} \Big)\Big)^{{\theta}_{2,m}}\nonumber\\
&+{{\theta}_{3,m}}\Big)\nonumber\\
\mathrm{s.t.}\quad &A_{12}+t_{{v},{w}}+t_{{v},{w'}}\leq \beta_{T,m}T_{all}, \quad \forall m\\
&A_{13}+t_{{v},{w}}p_{{v},{w}}+t_{{v},{w'}}p_{{v},{w'}}\nonumber\\
&\leq \beta_{E,m}E_{all}, \quad \forall m\\
&\sum_{m=1}^M\beta_{T,m}=1, \quad\sum_{m=1}^M\beta_{E,m}=1,
\end{align}
\end{subequations}
where
\begin{align}
A_{11}=&\mathop{\sum\sum}\limits_{\substack{(u,j)\not\in \{(v,w),(v,w')\}}}t^{\diamond}_{{u},{j}}B/A\times \mathrm{log}_2\left(1+s^{\diamond}_jF_{{u},{j}}p^{\diamond}_{{u},{j}}\right), \nonumber \\
A_{12}=&\mathop{\sum\sum}\limits_{(u,j)\not\in \{(v,w),(v,w')\}} t^{\diamond}_{{u},{j}}+{Tr(\mathbf{D}^T\mathbf{E^{\diamond}})}/{v},\nonumber \\
A_{13}=&\mathop{\sum\sum}\limits_{\substack{(u,j)  \not\in \{(v,w),(v,w')\}}}t^{\diamond}_{{u},{j}}p^{\diamond}_{{u},{j}}+\epsilon\Theta(\mathbf{E}^{\diamond}).\nonumber
\end{align}

Furthermore, under $\{t^{\diamond}_{{u}, {j}}, p^{\diamond}_{{u},{j}}\}$ with $t^{\diamond}_{{v},{w}}=0$, it can be seen that $(t_{{v},{w}},p_{{v},{w}},t_{{v},{w'}},p_{{v},{w'}})=(0,p^{\diamond}_{{v},{w}},\tilde{t},\tilde{p})$ is optimal to problem in (38), then, $(0,p^{\diamond}_{{v},{w}},\tilde{t},\tilde{p})$ must satisfy constraints (38a) and (38b), that is:
\begin{equation}
\quad A_{12}+\tilde{t}\leqslant \beta_{T,m}T_{all},\quad A_{13}+\tilde{t}\tilde{p}\leqslant \beta_{E,m}E_{all}.
\end{equation}

Comparing (39) with (38a) and (38b), it is evident that $(t_{{v},{w}},p_{{v},{w}},t_{{v},{w'}},p_{{v},{w'}})=(\tilde{t},\tilde{p},0,p^{\diamond}_{{v},{w'}})$ is also feasible for the two constraints. Then, substituting the two solutions $(t_{{v},{w}},p_{{v},{w}},t_{{v},{w'}},p_{{v},{w'}})=(0,p^{\diamond}_{{v},{w}},\tilde{t},\tilde{p})$ and $(t_{{v},{w}},p_{{v},{w}},t_{{v},{w'}},p_{{v},{w'}})=(\tilde{t},\tilde{p},0,p^{\diamond}_{{v},{w'}})$ into the objective function of problem in (38), we can obtain:
\begin{equation}
{{\theta}_{1,m}}\left( A_{11}+\tilde{t}B/A\times \mathrm{log}_2\left(1+s^{\diamond}_{w'}F_{{v},{w'}}\tilde{p}\right) \right)^{{\theta}_{2,m}}+{{\theta}_{3,m}},
\end{equation}
and
\begin{equation}
{{\theta}_{1,m}}\left( A_{11}+\tilde{t}B/A\times \mathrm{log}_2\left(1+s^{\diamond}_{w} F_{{v},{w}}\tilde{p}\right) \right)^{{\theta}_{2,m}}+{{\theta}_{3,m}}.
\end{equation}

Since $(t_{{v},{w}},p_{{v},{w}},t_{{v},{w'}},p_{{v},{w'}})=(0,p^{\diamond}_{{v},{w}},\tilde{t},\tilde{p})$ is optimal to problem in (38), the following inequality is obtained from (40) and (41):
\begin{equation}
\mathrm{log}_2\left(1+s^{\diamond}_{w'}F_{{v},{w'}}\tilde{p}\right)-\mathrm{log}_2\left(1+s^{\diamond}_{w}F_{{v},{w}}\tilde{p}\right)>0.
\end{equation}
By using Jensens inequality, we have
\begin{align}
&\mathrm{log}_2\left(s^{\diamond}_{w'}F_{{v},{w'}}\tilde{p}-s^{\diamond}_{w} F_{{v},{w}}\tilde{p}\right)\nonumber\\
&>\mathrm{log}_2\left(1+s^{\diamond}_{w'} F_{{v},{w'}}\tilde{p}\right)-\mathrm{log}_2\left(1+s^{\diamond}_{w} F_{{v},{w}}\tilde{p}\right)>0,
\end{align}
which leads to $s^{\diamond}_{w'}F_{{v},{w'}}\tilde{p}>s^{\diamond}_{w}F_{{v},{w}}\tilde{p}$ and subsequently results in $s^{\diamond}_{w'}F_{{v},{w'}}>s^{\diamond}_{w}F_{{v},{w}}$, but it contradicts to $w=argmax_{l\in\mathcal{J}}{s}^{\diamond}_{l}F_{{u},{l}}$. Consequently, $t^{\diamond}_{{v},{w}}\neq 0$ at vertex $w=argmax_{l\in\mathcal{J}}{s}^{\diamond}_{l}F_{{u},{l}}$.

\subsection {Proof of Part (iii)}

To prove part (iii), we assume that there exists user $v$ stored samples of class $c$ for task $m$, such that $t^{\diamond}_{{v},{w'}}\neq 0$ at vertex $w' \neq argmax_{l\in\mathcal{J}}{s}^{\diamond}_{l}F_{{u},{l}}$. In addition, based on the part A of this Appendix, we have $t^{\diamond}_{{v},{w}}\neq 0$ at vertex $w= argmax_{l\in\mathcal{J}}{s}^{\diamond}_{l}F_{{u},{l}}$. Accordingly, given part (i) of \textbf{Theorem 1}, we also have $p^{\diamond}_{{v},{w'}}p^{\diamond}_{{v},{w}}\neq 0$.

Then, the partial Lagrangian of $\mathrm{P}7$ corresponding to $(t_{{u},{j}},p_{{u},{j}})$ under fixed $\mathbf{s}=\mathbf{s^{\diamond}}$ and $\mathbf{E}=\mathbf{E^{\diamond}}$ can be established as
\begin{equation}
    \begin{aligned}
&\emph{L}\Big( \left\{t_{{u},{j}},p_{{u},{j}} \right\} \left\{{\varepsilon}_{c,m},\kappa, \chi, {\varsigma}_{{u},{j}}, {\tau}_{{u},{j}}, {\xi}_{{u},{j}}\right\} \Big)\\
&=\Big({\theta}_{1,m}\alpha_m^{{\theta}_{2,m}}+{\theta}_{3,m}\Big)+{\varepsilon}_{c,m}\Big(\alpha_m-\sum_{u\in\mathcal{G}_{c,m}}\sum_{j=1}^{J}t_{{u},{j}} \\
& B/A\times\mathrm{log}_2\Big(1+s^{\diamond}_jF_{{u},{j}}p_{{u},{j}}\Big)\Big)
+\kappa\Big(\sum_{u=1}^U\sum_{j=1}^Jt_{{u},{j}}p_{{u},{j}}\\
&+\epsilon\Theta(\mathbf{E}^{\diamond})-\beta_{E,m}E_{all}\Big)+\chi\Big(\sum_{u=1}^U\sum_{j=1}^Jt_{{u},{j}}-\beta_{T,m}T_{all}\\
&+{Tr(\mathbf{D}^T\mathbf{E^{\diamond}})}/{v}\Big)+\sum_{u=1}^U\sum_{j=1}^J{\varsigma}_{{u},{j}}\Big(\Big(1-s^{\diamond}_{j}\Big)t_{{u},{j}}\Big)\\
&-\sum_{u=1}^U\sum_{j=1}^J{\xi}_{{u},{j}}t_{{u},{j}}-\sum_{u=1}^U\sum_{j=1}^J{\tau}_{{u},{j}}p_{{u},{j}},
     \end{aligned}
\end{equation}
where $\left\{{\varepsilon}_{c,m},\kappa, \chi, {\varsigma}_{{u},{j}}, {\tau}_{{u},{j}}, {\xi}_{{u},{j}}\right\}$ are Lagrange multipliers. Because $\mathrm{P}7$ is convex in $\{{t}^{\diamond}_{{u},{j}}\}$ with fixed $\{{p}^{\diamond}_{{u},{j}}\}$ and vice visa, based on Karush-Kuhn-Tucker condition \cite{kkt}, the optimal primal-dual point $\{{t}^{\diamond}_{{u},{j}},{p}^{\diamond}_{{u},{j}}\}$ and $\left\{{\varepsilon}^{\diamond}_{c,m},\kappa^{\diamond}, \chi^{\diamond}, {\varsigma}^{\diamond}_{{u},{j}}, {\tau}^{\diamond}_{{u},{j}}, {\xi}^{\diamond}_{{u},{j}}\right\}$ must satisfy:
\begin{subequations}
\begin{align}
&{\tau}^{\diamond}_{{u},{j}}{p}^{\diamond}_{{u},{j}}=0,\quad {\xi}^{\diamond}_{{u},{j}}{t}^{\diamond}_{{u},{j}}=0,\quad \forall u,j\\
&{{\varepsilon}^{\diamond}_{c,m}}B/A\times log_2\left(1+s^{\diamond}_jF_{{u},{j}}p^{\diamond}_{{u},{j}}\right)+{\xi}^{\diamond}_{{u},{j}} \nonumber\\
&=\kappa^{\diamond}{p}^{\diamond}_{{u},{j}}+{\varsigma}^{\diamond}(1-s^{\diamond}_{j})+{\chi}^{\diamond}, \forall j\in\mathcal{J}, \quad \forall u\in\mathcal{G}_{c,m}\\
&-{{\varepsilon}^{\diamond}_{c,m}}\left(\frac{B{t}^{\diamond}_{{u},{j}}s^{\diamond}_j F_{{u},{j}}}{A\left(1+s^{\diamond}_j\times F_{{u},{j}}{p}^{\diamond}_{{u},{j}} \right)ln2}\right)+\kappa^{\diamond}{t}^{\diamond}_{{u},{j}}\nonumber\\
&-{\tau}^{\diamond}_{{u},{j}}=0, \forall j\in\mathcal{J}, \quad \forall u\in\mathcal{G}_{c,m}.
\end{align}
\end{subequations}

Putting ${t}^{\diamond}_{{v},{w'}},{t}^{\diamond}_{{v},{w}},{p}^{\diamond}_{{v},{w'}},{p}^{\diamond}_{{v},{w}}$ into (45a), it follows that
\begin{equation}
{\tau}^{\diamond}_{{v},{w'}}={\tau}^{\diamond}_{{v},{w}}={\xi}^{\diamond}_{{v},{w'}}={\xi}^{\diamond}_{{v},{w}}=0.
\end{equation}
Plugging ${\tau}^{\diamond}_{{v},{w'}}=0$ from (46) into (45c), the ${p}^{\diamond}_{{v},{w'}}$ can be obtained by
\begin{equation}
{p}^{\diamond}_{{v},{w'}}=\frac{{{\varepsilon}^{\diamond}_{c,m}}B}{A\kappa^{\diamond}ln2}-\frac{1}{s^{\diamond}_{w'}{F}_{{v},{w'}}}.
\end{equation}
Since when ${t}^{\diamond}_{{v},{w'}}\neq 0$, ${s}^{\diamond}_{w'}= 1$ from (28e), putting ${\xi}^{\diamond}_{{v},{w'}}=0$ from (46) and ${s}^{\diamond}_{w'}= 1$ into (45b), the following formula can be obtained
\begin{equation}
{{{\varepsilon}^{\diamond}_{c,m}}B/A\times}log_2\left({F}_{{v},{w'}}\frac{{{\varepsilon}^{\diamond}_{c,m}}B}{A\kappa^{\diamond}ln2}\right)-\frac{{{\varepsilon}^{\diamond}_{c,m}}B}{Aln2}+\frac{\kappa^{\diamond}}{{F}_{{v},{w'}}}=\chi^{\diamond}.
\end{equation}
With simple manipulations, (48) can be re-expressed as $F\left({s}^{\diamond}_{w'}{F}_{{v},{w'}}\right)=\chi^{\diamond}$,
where
\begin{equation}
F(x)={{\varepsilon}^{\diamond}_{c,m}}B/A\times log_2\left(x\frac{{{\varepsilon}^{\diamond}_{c,m}}B}{A\kappa^{\diamond}ln2}\right)-\frac{{{\varepsilon}^{\diamond}_{c,m}}B}{Aln2}+\frac{\kappa^{\diamond}}{x},
\end{equation}
with $x\neq 0$ (as ${s}^{\diamond}_{w'}= 1, {s}^{\diamond}_{w'}{F}_{{v},{w'}}\neq 0$).

With a similar procedure, by using ${\tau}^{\diamond}_{{v},{w}}={\xi}^{\diamond}_{{v},{w}}=0$ from (46), we have $F\left({s}^{\diamond}_{w'}{F}_{{v},{w}}\right)=\chi^{\diamond}$. As a results,
\begin{equation}
F\left({s}^{\diamond}_{w'}{F}_{{v},{w'}}\right)=F\left({s}^{\diamond}_{w}{F}_{{v},{w}}\right).
\end{equation}

Deriving the first-order derivative of $F(x)$ and combining with (47), we have
\begin{equation}
\nabla_x(F)=\frac{1}{x}\left(p^{\diamond}_{{v},{w'}}\kappa^{\diamond}+\kappa^{\diamond}\left(\frac{1}{{s}^{\diamond}_{w'}{F}_{{v},{w'}}}-\frac{1}{x}\right)\right).
\end{equation}

Since $p^{\diamond}_{{v},{w'}}>0$, it is clear that $\nabla_x(F)>0$ for any $x\in[{s}^{\diamond}_{w'}{F}_{{v},{w'}},{s}^{\diamond}_{w}{F}_{{v},{w}}]$ and $F(x)$ is a monotonically increasing function of $x$ in the interval. By combining the conclusion with (50), we can obtain that
${s}^{\diamond}_{w'}{F}_{{v},{w'}}={s}^{\diamond}_{w}{F}_{{v},{w}}=argmax_{l\in\mathcal{J}}{s}^{\diamond}_{l}F_{{u},{l}}$. But this contradicts to $w' \neq argmax_{l\in\mathcal{J}}{s}^{\diamond}_{l}F_{{u},{l}}$. Consequently, $t^{\diamond}_{{v},{w'}}= 0$ at vertex $w' \neq argmax_{l\in\mathcal{J}}{s}^{\diamond}_{l}F_{{u},{l}}$.

\section{Proof of Proposition 1}

In this appendix, we prove the joint-convexity of problem $\mathrm{P}6$. Since $\nabla^{2}f(\alpha_m)\prec 0, \forall \alpha_m>0$, the objective function ${\theta}_{1,m}\alpha_m^{{\theta}_{2,m}}+{\theta}_{3,m}$ is concave with respect to $\alpha_m$.

To verify the convexity of the constraint (25a), we first define a new affine function as:
\begin{align}\label{4}
f_1(\delta_{{u}})=e_{{u}}+\delta_{u}M_u(s^{\diamond}).
\end{align}
Since $log_2\left({1}/{f_1(\delta_{u})}\right)$ is a convex function of $\delta_{u}$ \cite{kkt}, then its perspective function $e_{u}\times B/A\mathrm{log}_2\left({e_{u}}/{f_1}(\delta_{{u}}) \right)$ is convex with respective to $e_{{u}}$ and ${\delta}_{{u}}$, and consequently
(25a) is jointly convex with respect to $\alpha_m$, $e_{{u}}$ and ${\delta}_{{u}}$ as its Hessian matrix is positive semi-definite.

Besides, since constraints (25b) and (25c) are affine, both of them are convex. Therefore, the objective function is concave and all the constraints are convex, constituting a joint convex optimization problem $\mathrm{P}6$.

\end{document}